\newif\ifreport\reporttrue
\documentclass[sigconf]{acmart}

\pdfoutput=1

\renewcommand\footnotetextcopyrightpermission[1]{} 

\setcopyright{acmcopyright}
\copyrightyear{2018}
\acmYear{2018}
\acmDOI{10.1145/1122445.1122456}
\acmConference[Woodstock '18]{Woodstock '18: ACM Symposium on Neural
  Gaze Detection}{June 03--05, 2018}{Woodstock, NY}
\acmPrice{15.00}
\acmISBN{978-1-4503-XXXX-X/18/06}
    
\usepackage{amsmath}
\usepackage{amsthm}
\usepackage{amsfonts} 
\usepackage{float}
\usepackage{graphicx}
\usepackage{stfloats}
\usepackage{subfig}
\usepackage{float}
\usepackage{bbm}
\usepackage[linesnumbered,ruled]{algorithm2e}
\usepackage{cases}

\DeclareMathOperator*{\argmin}{arg\,min}

\title{Battle between Rate and Error in Minimizing Age of Information}

\begin{document}

\author{Guidan Yao}
\affiliation{The Ohio State University}
\email{yao.539@osu.edu}

\author{Ahmed M. Bedewy}
\affiliation{The Ohio State University}
\email{bedewy.2@osu.edu}
\author{Ness B. Shroff}
\affiliation{The Ohio State University}
\email{shroff.11@osu.edu}
%
\begin{abstract}
In this paper, we consider a status update system, in which update packets are sent to the destination via a wireless medium that allows for multiple rates, where a higher rate also naturally corresponds to a higher error probability. The data freshness is measured using age of information, which is defined as the age of the recent update at the destination. A packet that is transmitted with a higher rate, will encounter a shorter delay and a higher error probability. Thus, the choice of the transmission rate affects the age at the destination. In this paper, we design a low-complexity scheduler that selects between two different transmission rate and error probability pairs to be used at each transmission epoch. This problem can be cast as a Markov Decision Process. We show that there exists a threshold-type policy that is age-optimal. More importantly, we show that the objective function is quasi-convex or non-decreasing in the threshold, based on to the system parameters values. This enables us to devise a \emph{low-complexity algorithm} to minimize the age. These results reveal an interesting phenomenon: While choosing the rate with minimum mean delay is delay-optimal, this does not necessarily minimize the age. 
\end{abstract}

%

\keywords{Age of Information, Markov Decision Process, Transmission Rate, Threshold Policy}

\maketitle
\pagestyle{plain}
\section{Introduction}
\emph{Age of information} is a new metric that has attracted significant recent attention \cite{yates2015lazy,kadota2018scheduling,bedewy2020optimizing}. This concept has been motivated by the rapid growth of real-time applications, e.g., health monitoring, automatic driving system, and intelligent agriculture, etc. For such applications, freshness of information updates is of utmost importance. However, traditional metric like delay cannot fully characterize the freshness of information updates. For example, if information is updated infrequently, then the updates are not fresh even though the delay is small. To this end, age of information, or simply age, was proposed in \cite{kaul2011minimizing} as a measure of the data freshness. Specifically, age of information is defined as the time elapsed since the generation of the most recently received status update. 

There exist many works dealing with the age minimization problem. 
One class of works have focused on investigating optimal sampling and updating policy to minimize age of information. In \cite{sun2017update}, authors study the updating policy to minimize age in the presence of queuing delay. In \cite{bacinoglu2015age, wu2017optimal,zhou2018optimal,ceran2019average}, sampling and updating polices are studied under energy constraint. In \cite{bacinoglu2015age, wu2017optimal}, the authors assume that the channel is noiseless while in \cite{zhou2018optimal}, authors assume that channel state is known a priori and updating cost is a function of channel state to ensure successful transmission. 
In \cite{ceran2019average}, the authors consider transmission failure and investigate optimal sampling policy for age minimization under energy constraint. 
These works consider the effects of queueing delay, channel state, energy supply and minimize the age of information by controlling sampling and updating times, in which case they assume that there is only one transmission mode to transmit updates. However, in real systems, updates can be sent to a destination using heterogenous transmissions in terms of transmission delay and error probability. Two examples are provided as follows:

\emph{Error rate control:} 
Error rate control scheme is managed at physical layer.
 In particular, the transmission rate is often adapted via modulation and coding scheme to meet a fixed target error rate \cite{du2020balancing}. It is known that choosing a lower target error rate corresponds to a lower transmission rate, and hence a longer transmission delay. On the other hand, a higher transmission rate (i.e., a lower transmission delay) also corresponds to a higher transmission error probability of information delivery. Thus, there is a tradeoff between transmission delay and transmission success probability, both of which are affected by the target error rate.

\emph{Scheduling over channels in different frequencies:} It is common that a device can access channels in different frequencies. For example, cellphones can access WiFi (high frequency) and LTE (low frequency). If updates are transmitted over such devices, then the age of information may experience different transmission properties based on the carrier frequency. 
In particular, it is known that it is hard for radio waves to distract obstacles that are in same or larger size than their wavelength. Thus, low-frequency radios (longer wavelength) are less vulnerable to blockage than high-frequency radios, which implies that low frequency channels are more reliable than their high frequency counterparts. Of course, the higher frequency channels allow for higher rate (lower delay) transmissions, resulting in a similar tradeoff between the transmission delay and transmission success probability.  

The above examples clearly indicate that, transmission of updates can experience different transmission delays and error probabilities based on the choice of either target error rate or carrier frequency. In particular, a decrease in the transmission error probability will increase the chances of a successful update delivery (decrease age) while an increase in the transmission delay will increase the inter-delivery time (increase age). Thus, the key questions are: \emph{when is it optimal to use the lower transmission rate with a lower error probability?};\emph{which variable plays a more important role in determining the optimal actions?}. To address these questions, we begin by investigating a status update system with two heterogenous transmissions and obtain the optimal transmission selection policy to minimize the average age. 
Studying the two-rate scenario provides us with some insights in the optimal policy for a more general multi-rate (multi-error probability) scenario, which is discussed in Section \ref{diss}, and provides basis for our future work. Specifically, our contributions are outlined as follows:

\begin{itemize}
	\item We investigate the optimal trade-off between transmission delay and error probability for minimizing age. We show that there exists a stationary deterministic optimal transmission selection policy. Moreover, we show that the optimal transmission selection policy is of threshold-type in terms of age (Theorem \ref{THRE_PROP}). In particular, we show that the optimal decision is non-increasing (non-decreasing) in age if the mean delay of the low rate transmission is smaller (larger) than that of the high rate transmission. This result was not anticipated: For example, in \cite{ozkan2014optimal,de2005managing}, it was shown that the optimal delay policy chooses the server with minimum mean delay whenever it is available. With this, one may expect that using the transmission with higher mean delay would worsen the age performance. Surprisingly, however, we show that choosing the transmission with higher mean delay can sometimes improve the age performance. 
	\item We derive the average cost as a function of the threshold with the aid of the state transition diagram. We then optimize the threshold to minimize the average cost function. In particular, although the optimization problem is non-convex, we are able to show that if the mean delay of the low rate transmission is smaller than that of the high rate transmission, the objective function is quasi-convex; otherwise, the optimal policy chooses higher rate transmission (Theorem \ref{CONVEXITY}). This enables us to devise a low-complexity algorithm to obtain the optimal policy.  
\end{itemize}

The remainder of this paper is organized as follows. The system model is introduced in Section \ref{sec2}. In Section \ref{sec3}, we map the problem to an equivalent problem which can be regarded as an average cost MDP, and then formulate the MDP problem. In Section \ref{structure}, we explore the structure of the optimal policy and properties of average cost function, and devise an efficient algorithm. In Section \ref{diss}, we provide a disscusion on multi-rate scenario. In Section \ref{sec5}, we provide numerical results to verify our theoretical results. 

\section{System Model}
\label{sec2}
\begin{figure}
	\centering
	\includegraphics[width=0.44\textwidth]{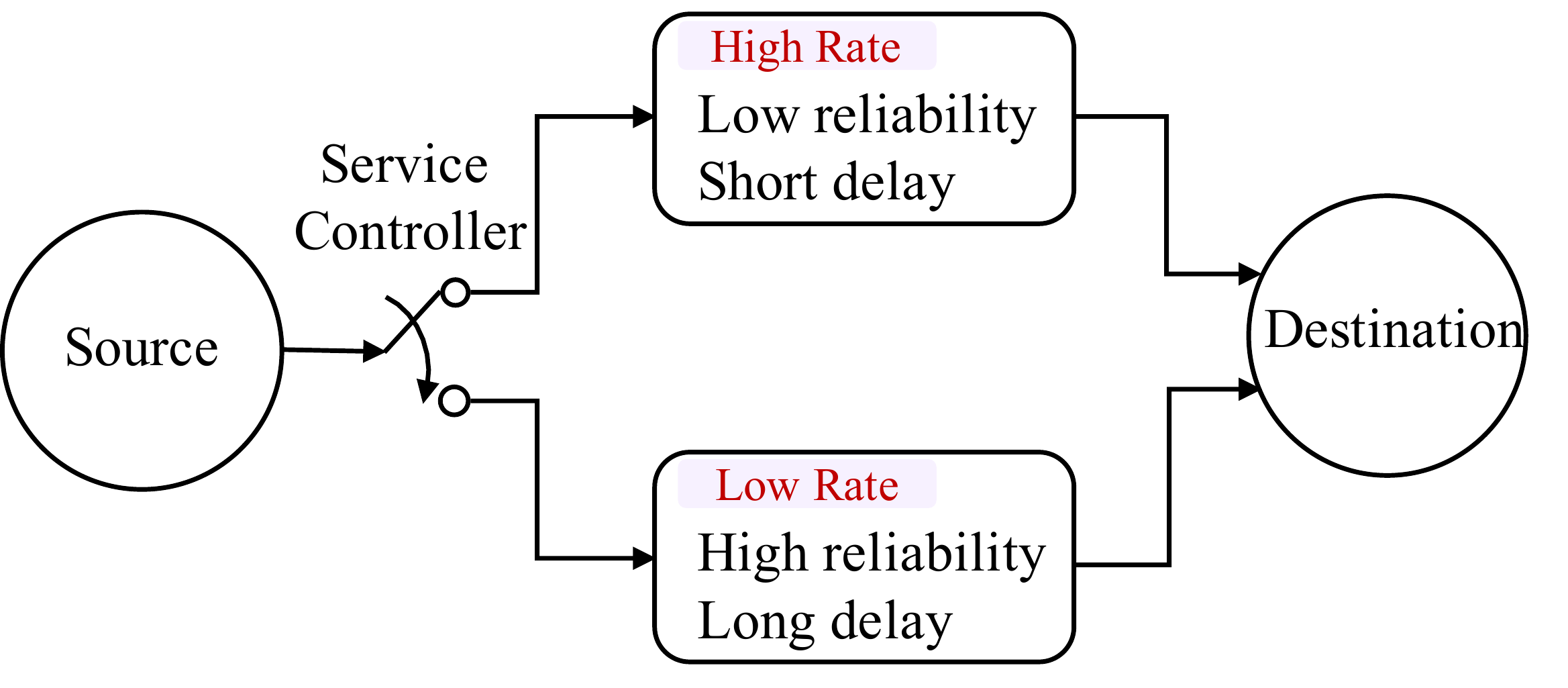}
	\caption{Model}
	\label{model}
\end{figure}
We consider a status update system, in which update packets are sent to the destination via a wireless medium with varying transmission delay and error probability. The update packets are generated whenever the wireless medium becomes idle.  
We assume that there are two heterogenous transmissions available for updating, namely low rate and high rate transmissions. The high rate transmission offers a shorter transmission delay than low rate transmission; while low rate transmission offers more reliable transmission than high rate transmission. A decision maker chooses a transmission rate for each transmission opportunity. 
We denote the set of transmission rates as $\mathcal{U}\triangleq\{1,2\}$, where 1 and 2 denote the low rate and high rate transmissions, respectively. We use $\mathcal{P}\triangleq\{p_j:0<p_j<1,j\in\mathcal{U}\}$ and $\mathcal{D}\triangleq\{d_j: 0<d_j<\infty,j\in\mathcal{U}\}$ to denote the set of transmission error probabilities and transmission delays, respectively. Transmission $j\in\mathcal{U}$ corresponds to transmission delay $d_j$ and transmission error probability $p_j$. We assume that $d_1>d_2$ and $p_1<p_2$.

We use $Y_i$ to denote the transmission delay of packet $i$, where $Y_i\in \mathcal{D}$. Let $D_i$ denote the delivery time of packet $i$. Since updates are generated whenever the wireless medium becomes idle, $D_i$ equals to the generation time of packet $i+1$. Also, we have $D_i=\sum_{j=1}^{i}Y_j$. 

\begin{figure}[]
	\centering
	\includegraphics[width=.35\textwidth]{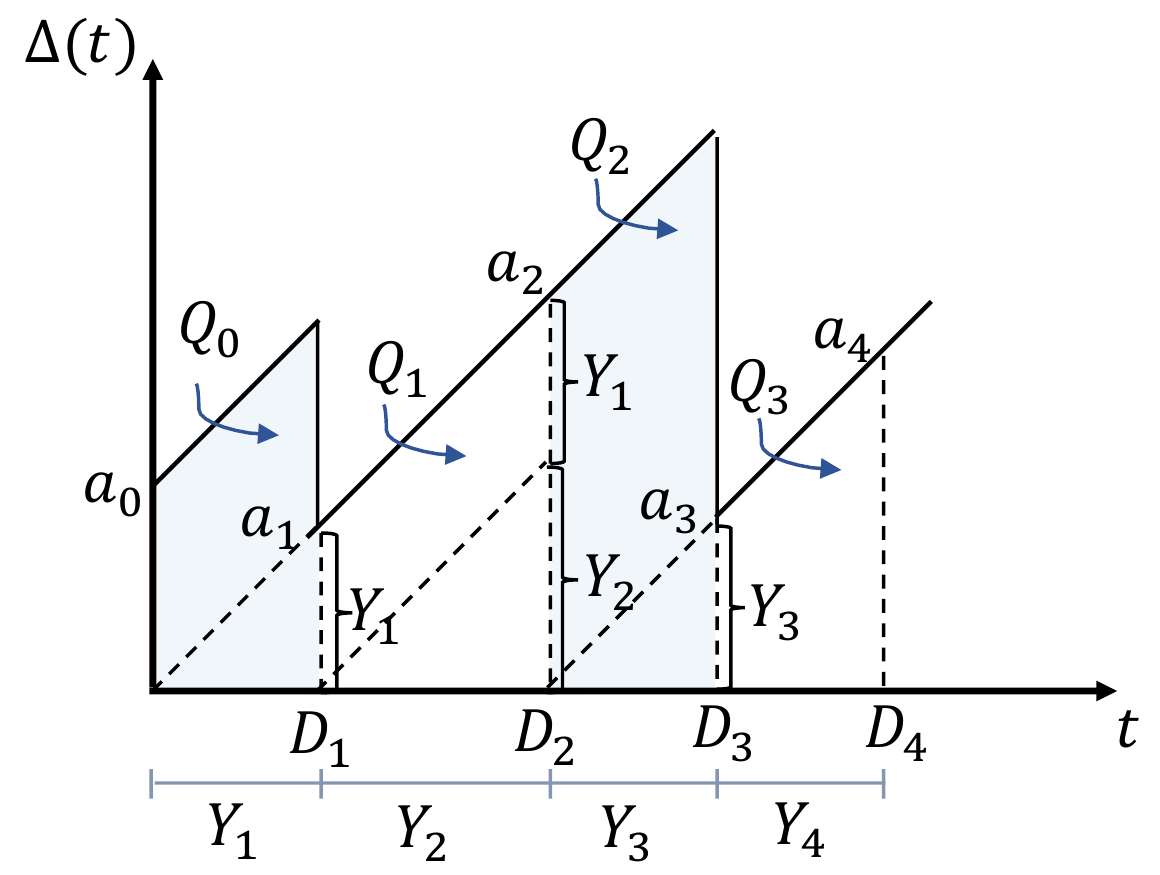}
	\caption{Example of Age Evolution}
	\label{fig1}
\end{figure}

At any time $t$, the most recently received update packet is generated at time 
\begin{equation}
	U(t)=\max\{D_i:D_{i+1}\leq t\}.
\end{equation}
Then, the \emph{age of information}, or simply \emph{age} is defined as
\begin{equation}
	\Delta(t)=t-U(t).
\end{equation}

The age $\Delta(t)$ is a stochastic process that increases with $t$ between update packets and is reset to a smaller value upon the successful delivery of a fresher packet. We suppose that the age $\Delta(t)$ is right-continuous. As shown in Fig. \ref{fig1}, packet $2$ is sent at time $D_1$ and its delivery time is $D_2=D_1+Y_2$. Since this packet transmission fails, the age does not reset to a smaller value at $D_2$. Packet $3$ transmission starts at $D_2$, which is successfully delivered at time $D_3$. Thus, the age increases linearly until it reaches to $\Delta(D_3^-)=Y_1+Y_2+Y_3$ before packet $3$ is successfully sent, and then drops to $\Delta(D_3)=Y_3$ at $D_3$.

\section{Optimization Problem}
\label{sec3}
We use $u_i$ to denote which transmission rate (low or high) is selected to transmit packet $i$, where $u_i \in \mathcal{U}$. In particular, if $u_i=1$ (or $u_i=2$), then packet $i$ is transmitted using the low (or high) rate transmission, encounters transmission delay $d_{1}$ (or $d_{2}$), and is received successfully with probability $1-p_{1}$ (or $1-p_{2}$). A transmission selection policy $\pi$ specifies a transmission selection decision for each stage\footnote[1]{Stage $i$ corresponds to the duration from $D_{i-1}$ to $D_{i}$.}. For any policy $\pi$, we define the total average age as
\begin{equation}
	\bar{\Delta}(\pi)=\limsup_{n \rightarrow \infty}\frac{\mathbb{E}[\int_{0}^{D_n}\Delta(t)dt]}{\mathbb{E}[D_n]}.
\end{equation}

Our goal is to seek a transmission selection policy that solves the total average age minimization problem as follows:
\begin{align}
	\bar{\Delta}^*=\min_{\pi \in \Pi} \bar{ \Delta}(\pi),
	\label{ori_prob}
\end{align}
where $\bar{\Delta}^*$ denotes the optimal total average age. Let $\Pi$ denote the set of all causal transmission selection policies, in which the policy $\pi\in\Pi$ depends on the history and current system state.

\subsection{Equivalent Mapping of Problem \eqref{ori_prob}}
We decompose the area under the curve $\Delta(t)$ into a sum of disjoint geometric parts as shown in Fig. \ref{fig1}. Observing the area in interval $[0, D_n]$, the area can be regarded as the concatenation of the areas $Q_i$. Then,
\begin{align}
	\int_{0}^{D_n} \Delta (t) dt=\sum_{i=0}^{n-1} [Q_i].\label{map}
\end{align}
Let $a_i$ denote the age at time $D_i$, i.e., $a_i=\Delta(D_i)$. Then, $Q_i$ can be expressed as
\begin{align}
	Q_i=a_i Y_{i+1}+\frac{1}{2}Y_{i+1}^2.\label{q}
\end{align}
Recall that $D_n=\sum_{i=0}^{n-1} Y_{i+1}$. With Eq. \eqref{q} and Eq. \eqref{map}, the total average age is expressed as
\begin{align}
\limsup_{n\rightarrow \infty}	\frac{\sum_{0}^{n-1} \mathbb{E}[a_i Y_{i+1}+\frac{1}{2}Y_{i+1}^2]}{\sum_{i=0}^{n-1}\mathbb{E}[ Y_{i+1}]}..
\label{total_avg_age}
\end{align}

With this, the optimal transmission selection problem for minimizing the total average age can be formulated as
\begin{align}
	\bar{\Delta}^*\triangleq\min_{\pi \in \Pi} \limsup_{n\rightarrow \infty}	\frac{\sum_{0}^{n-1} \mathbb{E}[a_i Y_{i+1}+\frac{1}{2}Y_{i+1}^2]}{\sum_{i=0}^{n-1}\mathbb{E}[ Y_{i+1}]}.
	\label{equiv1}
\end{align}
The problem is hard to solve in current form. Thus, we provide an equivalent mapping for it. A problem with parameter $\beta$ is defined as follows:
\begin{align}
	p(\beta)\triangleq\min_{\pi \in \Pi} \limsup_{n\rightarrow \infty}	\frac{1}{n}\sum_{0}^{n-1} \mathbb{E}[(a_i-\beta) Y_{i+1}+\frac{1}{2}Y_{i+1}^2]\label{equiv2}.
\end{align}

\begin{lemma}
\label{TRANS}
The following statements are true:	

\noindent(i) $\bar{\Delta}^* \gtreqqless \beta$ if and only if $p(\beta) \gtreqqless 0$;

\noindent(ii) If $p(\beta)=0$, then the optimal transmission selection policies that solve \eqref{equiv1} and \eqref{equiv2} are identical.
\end{lemma}
\ifreport
\begin{proof}
See Appendix \ref{appe0}.
\end{proof}
\fi

By Lemma \ref{TRANS}, if $\beta=\bar{\Delta}^*$, then the optimal transmission selection policies that solve \eqref{equiv1} and \eqref{equiv2} are identical. With this, given $\beta$, we formulate the problem \eqref{equiv2} as an infinite horizon average cost per stage MDP in Section \ref{mdp} and show that the optimal policy for \eqref{equiv2} is of threshold-type in Section \ref{struc}. Since the value of $\beta$ is arbitrary, we will be able to conclude that the optimal policy for \eqref{equiv1} is of threshold-type. In addition, in Section \ref{alg}, we are able to devise a low-complexity algorithm to obtain the optimal threshold.

\subsection{The MDP problem of \eqref{equiv2}}
\label{mdp}
From \cite{bertsekas1995dynamic}, given $\beta$, Problem \eqref{equiv2} is equivalent to an average cost per stage MDP problem. The components of the MDP problem are described as follows: 

\begin{itemize}
\item \textbf{States:} The system state at stage $i$ is the age $a_i$. In this paper, we consider the state space $\mathcal{S}\triangleq \{a=ld_1+vd_2: l, v\in \{0,1,\cdots\}\}$. If the initial state is outside $\mathcal{S}$, then eventually the state will enter $\mathcal{S}$ (with state $d_1$ or $d_2$); otherwise, a successful packet transmission never occurs. In fact, the maximal probability that no transmission succeeds after $l$ stages is $p_2^l$, which decreases with number of stages $l$. After state enters $\mathcal{S}$, it will stay in $\mathcal{S}$ onwards (since transmission delay is either $d_1$ or $d_2$). Note that $\mathcal{S}$ is unbounded since successful packet transmission happens with certain probabilities.  

\item \textbf{Actions:} At delivery time $D_{i-1}$, the action that is chosen for stage $i$ is $u_i \in \mathcal{U}$. The action $u_i$ determines the transmission delay. For example, if $u_i=1$, then the transmission delay at stage $i$ is $d_{1}$. 

\item \textbf{Transition probabilities:} Given the current state $a_i$ and action $u_i$ at stage $i$, the transition probability to the state $a_{i+1}$ at the stage $i+1$ is defined as
\begin{align}
 &P(a_{i+1}=a'|a_i=a, u_i=u)=
  \begin{cases}
  p_{u} & \text{if}\, \, a'=a+d_u, \\
  1-p_{u} & \text{if} \, \, a'=d_u,\\  
  0 & \text{otherwise}.
  \end{cases} 
\end{align}

\item \textbf{Costs:} Given state $a_i$ and action $u_i$ at stage $i$, the cost at the stage is defined as
\begin{align}
	C(a_i,u_i)=(a_i-\beta)d_{u_i}+\frac{1}{2}d_{u_i}^2.\label{inst_cost}
\end{align}

\end{itemize}

Given $\beta$, the average cost per stage under a transmission selection policy $\pi$ is given by
\begin{align}
	&J(\pi,\beta)\triangleq\limsup_{n \rightarrow \infty} \frac{1}{n}\mathbb{E_\pi}\left[\sum_{i=0}^{n-1} C(a_i,u_i)\right].
	\label{avg_cost}
\end{align}
Our objective is to find a transmission selection policy $\pi\in\Pi$ that minimizes the average cost per stage, which can be formulated as 

\emph{Problem 1 (Average cost MDP)}
\begin{align}
	\min_{\pi\in\Pi} J(\pi,\beta).
		\label{avg_prob}
\end{align}
We say that a transmission selection policy $\pi$ is \emph{average-optimal} if it solves the problem in \eqref{avg_prob}. Our goal is to find the average-optimal policy. A transmission selection policy is a sequence of decision rules, i.e., $\pi=(\zeta_1,\zeta_2,\cdots)$, where a decision rule $\zeta_i$ maps the history of states and actions, and the current state to an action.
A transmission selection policy is called a stationary deterministic policy if $u_i\!=\!\zeta(a_i)$ for all $i\in\mathbb{N}$, where $\zeta: \!\mathcal{S}\!\rightarrow \mathcal{U}$ is a deterministic function. 
Stationary deterministic policies are the easiest to be implemented and evaluated. However, there may not exist a stationary deterministic policy that is average-optimal \cite{bertsekas1995dynamic}. Next, we show that there exists a stationary deterministic transmission selection policy that is average-optimal. Moreover, we show that the optimal policy is of a threshold-type.

\section{Structure of Average-Optimal Policy and Algorithm Design}
\label{structure}
In this section, we investigate the structure of average-optimal policy that minimizes the average cost in \eqref{avg_cost} and propose an efficient algorithm.
\subsection{Threshold Structure of Average-Optimal Policy}
\label{struc}
\subsubsection{Threshold structure:}
The following theorem states that there exists a threshold-type stationary deterministic policy that is average-optimal. In particular, the problem is divided into two cases based on the relation between $d_1(1-p_2)$ and $d_2(1-p_1)$. Under these two cases, the threshold-type average-optimal policy shows opposite behaviors. 
\begin{theorem}
\label{THRE_PROP}
There exists a stationary deterministic average-optimal transmission selection policy that is of threshold-type. Specifically,

\noindent(i) If $d_1(1-p_2)\leq d_2(1-p_1)$, then \eqref{avg_cost} can be minimized by the policy of the form $\pi^*=(\zeta^*,\zeta^*,\cdots)$, where
\begin{align}
\label{type2}
 &\zeta^*(a)=
  \begin{cases}
  2 & \text{if}\, \, 0\leq a\leq a_1^*, \\
  1 & \text{if} \, \,a_1^*< a,
  \end{cases} 
\end{align}
where $a_1^*$ denotes the age threshold.

\noindent(ii) If $d_1(1-p_2)\geq d_2(1-p_1)$, then \eqref{avg_cost} can be minimized by the policy of the form $\pi^*=(\zeta^*,\zeta^*,\cdots)$, where
\begin{align}
\label{type1}
 &\zeta^*(a)=
  \begin{cases}
  1 & \text{if}\, \, 0\leq a\leq a_2^*, \\
  2 & \text{if} \, \,a_2^*< a,
  \end{cases} 
\end{align}
where $a_2^*$ denotes the age threshold.

\end{theorem}
\begin{proof}
	Please see Section \ref{proof}.
\end{proof}

Define the mean delay of transmission $j\in\mathcal{E}$ as 
\begin{align}
	\bar{d}_j\triangleq\frac{d_j}{1-p_j}.
	\label{mean_delay}
\end{align}
By Theorem \ref{THRE_PROP} (i), when the age exceeds a certain threshold, the optimal policy chooses the transmission with smaller mean delay. This result reveals an interesting phenomenon: While the transmission with minimum mean delay is the optimal decision for minimizing the average delay, this does not necessarily minimize the age.  
In particular, when the age is below a certain threshold, the average age is reduced by choosing a faster transmission that has a higher mean delay (i.e., a higher error probability). The reason is that if successful, the age remains low. If it fails, it provides an opportunity to generate a later packet that can be transmitted in a shorter period of time.
In Section \ref{alg}, based on Theorem \ref{THRE_PROP} (ii), we will show that the optimal policy under $d_1(1-p_2)\geq d_2(1-p_1)$ is to choose transmission rate $2$ for each transmission opportunity. This is reasonable because both the delay and mean delay (including the impact of the error probability) of transmission rate $2$ is shorter than that of transmission rate $1$.

\subsubsection{Proof of Theorem \ref{THRE_PROP}}
\label{proof}
One way to investigate the average cost MDPs is to relate them to the discounted cost MDPs. To prove Theorem \ref{THRE_PROP}, we (i) address a discounted cost MDP, i.e., establish the existence of a stationary deterministic policy that solves the MDP and then study the structure of the optimal policy; and (ii) extend the results to the average cost MDP in \eqref{avg_prob}.

Given an initial state $a$, the total expected $\alpha$-discounted cost under a transmission selection policy $\pi\in \Pi$ is given by 
\begin{align}
	V^\alpha(a;\pi)=\limsup_{n\rightarrow\infty}\mathbb{E}\left[\sum_{i=0}^{n-1}\alpha^i C(a_i,u_i)\right],
	\label{disc_cost}
\end{align}
where $0<\alpha<1$ is the discount factor. Then, the optimization problem of minimizing the total expected $\alpha$-discounted cost can be cast as

\emph{Problem 2 (Discounted cost MDP)} 
\begin{align}
	V^\alpha(a)\triangleq\min_{\pi} V^\alpha(a;\pi),
	\label{disc_prob}
\end{align} 
where $V^\alpha(a)$ denotes the optimal total expected $\alpha$-discounted cost. A transmission selection policy is said to be \emph{$\alpha$-discounted cost optimal} if it solves the problem in \eqref{disc_prob}. In Proposition \ref{DISC_EXIST}, we show that there exists a stationary deterministic transmission selection policy which is $\alpha$-discounted cost optimal and provide a way to explore the property of the optimal policy.

\begin{proposition}
\label{DISC_EXIST}
(a) The optimal total expected $\alpha$-discounted cost $V^\alpha$ satisfies the following optimality equation:
\begin{align}
	V^\alpha(a)=\min_{u\in \mathcal{U}} Q^\alpha(a, u),
	\label{bellman_equ}
\end{align}
where 
\begin{align}
	Q^\alpha(a, u)=C(a,u)+\alpha p_uV^\alpha(a+d_u)+\alpha(1-p_u)V^\alpha(d_u).
	\label{qvalue}
\end{align}

(b) The stationary deterministic policy determined by the right-hand-side of \eqref{bellman_equ} is $\alpha$-discounted cost optimal.

(c) Let $V_n^\alpha(a)$ be the cost-to-go function such that $V_0^\alpha(a)=\frac{d_1-d_2}{\alpha(p_2-p_1)}a$ and for $n\geq0$
\begin{align}
	V_{n+1}^\alpha(a)=\min_{u\in \mathcal{E}} Q_{n+1}^\alpha(a, u),
	\label{value_iter}
\end{align}
where
\begin{align}
	Q_{n+1}^\alpha(a, u)=C(a, u)+\alpha p_uV_n^\alpha(a+d_u)+\alpha(1-p_u)V_n^\alpha(d_u).
\end{align}
Then, we have that for each $\alpha$, $V_{n}^\alpha(a)\rightarrow V^\alpha(a)$ as $n\rightarrow \infty$.
\end{proposition}
 \ifreport
\begin{proof}
See Appendix \ref{proposition}.
\end{proof}
\fi

Next, with the optimality equation \eqref{bellman_equ} and value iteration \eqref{value_iter}, we show that the optimal policy is of threshold-type in Lemma \ref{THRESHOLD}.

\begin{lemma} \label{THRESHOLD}
Given a discount factor $\alpha$,

(i) if $(1-\alpha p_2)d_1\leq (1-\alpha p_1)d_2$, then the $\alpha$-discounted cost optimal policy is of threshold-type, i.e., the optimal decision is non-increasing in age.

(ii) if $(1-\alpha p_2)d_1\geq (1-\alpha p_1)d_2$, then the $\alpha$-discounted cost optimal policy is of threshold-type, i.e., the optimal decision is non-decreasing in age. 

\end{lemma}

\begin{proof}
Please see Appendix \ref{discount_threshold}
\end{proof}

This lemma proves that the $\alpha$-discounted cost optimal policy is of threshold type. Next, we extend the results to average cost MDP and show that there exists a stationary deterministic average-optimal policy which is of threshold-type. Based on the results in \cite{sennott1989average}, we have the following lemma, which provides a candidate for average-optimal policy.
\begin{lemma}
\label{CANDI}
	(i) Let $\alpha_n$ be any sequence of discount factors converging to 1 with $\alpha_n$-discounted cost optimal stationary deterministic policy $\pi^{\alpha_n}$. There exists a subsequence $\gamma_n$ and a stationary policy $\pi^\star$ that is a limit point of $\pi^{\gamma_n}$. 
	
	(ii) If $d_1(1-p_2)\leq d_2(1-p_1)$, $\pi^\star$ is of threshold-type in \eqref{type2}; if $d_1(1-p_2)\geq d_2(1-p_1)$, $\pi^\star$ is of threshold-type in \eqref{type1}.
\end{lemma}
\ifreport
\begin{proof}
See Appendix \ref{candidate_policy}.
\end{proof}
\fi

\ifreport
By \cite{sennott1989average}, under certain conditions (A proof of these conditions verification is provided in Appendix \ref{app1}), $\pi^\star$ is average-optimal.  
\fi

\subsection{Algorithm Design}
\label{alg}
Recall that if $p(\beta)=0$, then the optimal transmission selection policies that solve \eqref{equiv1} and \eqref{equiv2} are identical. Given $\beta$, the optimal policy that solves \eqref{equiv2} is of threshold-type by Theorem \ref{THRE_PROP} and then \eqref{equiv2} can be re-expressed as
\begin{numcases}{p(\beta)\triangleq}
\min_{\pi \in \Pi_1} J(\pi,\beta), & if $d_1(1-p_2)< d_2(1-p_1)$\label{pBeta2}\\
\min_{\pi \in \Pi_2} J(\pi,\beta), & if $d_1(1-p_2)\geq d_2(1-p_1)$\label{pBeta1}
\end{numcases}
where $\Pi_1$ and $\Pi_2$ denote the sets of threshold-type policies in \eqref{type2} and \eqref{type1}, respectively. 
Thus, the optimal policy that solves \eqref{equiv1} can be obtained with two steps: 
\begin{itemize}
	\item \textbf{Step (i):} For each $\beta$, find the $\beta$-associated optimal policy $\pi^*_\beta$ such that $p(\beta)=J(\pi^*_\beta,\beta)$. 	
\item \textbf{Step (ii):} Find $\beta^*$ such that $p(\beta^*)=0$. This implies $\pi^*_{\beta^*}$ solves \eqref{equiv1}. 
\end{itemize}
To narrow our searching range in (ii), in Lemma \ref{BETA_RANGE}, we provide a lower bound $\beta_\text{min}$ and an upper bound $\beta_\text{max}$ of $\beta^*$. Then, for (i), we only need to pay attention to $p(\beta)$ for $\beta\in[\beta_\text{min},\beta_\text{max}]$. 

In particular, within the range of $\beta$, we show that $J$ in \eqref{pBeta2} is \emph{quasi-convex} in a threshold related variable, which enables us to devise a \emph{low-complexity} algorithm based on golden section search. Moreover, we show that $\pi^*_\beta$ that solves \eqref{pBeta1} always chooses $u=2$, which allows us to get the optimal policy for \eqref{equiv1} directly. 
 
\begin{lemma}
\label{BETA_RANGE}
	The parameter $\beta^*$ is lower bounded by $\beta_\text{min}\triangleq 1.5d_2$ and upper bounded by $\beta_\text{max}\triangleq\min\Big\{(\frac{1}{1-p_1}+0.5)d_1,(\frac{1}{1-p_2}+0.5)d_2\Big\}$.
\end{lemma}
\ifreport
\begin{proof}
See Appendix \ref{range}.
\end{proof}
\fi

In the following content, we provide a theoretical analysis step by step for our algorithm design in Algorithm \ref{alg1}, which returns the optimal threshold and optimal average age. 

\begin{algorithm}[]
\caption{Threshold-based Age-Optimal Policy}
\label{alg1}
\LinesNumbered
\footnotesize
given $d_1$, $d_2$, $p_1$, $p_2$,$\tau=(\sqrt{5}-1)/2$, tolerance $\epsilon_1$, $\epsilon_2$, $l=\beta_\text{min}$, $r=\beta_\text{max}$\;
\While {$r-l>\epsilon_1$}{	
 $\beta=\frac{r+l}{2}$\;
 \eIf{$d_1(1-p_2)\geq d_2(1-p_1)$}{
 $m=0$, $k=0$, $J^*=J_2(0,0,\beta)$\;}
 {
 $k_\text{max}=\lfloor\frac{d_1}{d_2}\rfloor$, $J^*=f(1,0,\beta)$\;
	\ForEach{ $k_1\in \{0,1,\cdots,k_\text{max}\}$ }{		
		\eIf{$\frac{\partial J_1(y, k_1,\beta)}{\partial y}\vert_{y=1}<0$}{
			$m=0$\;
			}{
		$y_0=0$, $y_1=1$, $y_2=y_1-(y_1-y_0)\tau$, $y_3=y_0+(y_1-y_0)\tau$\;
		\While{$p_2y_1 \geq y_0$ and $y_1-y_0>\epsilon_2$}{
			\eIf{$J_1(y_2,k_1,\beta)>J_1(y_3,k_1,\beta)$}{
			$y_0=y_2$\;
			}{
			$y_1=y_3$\;
			}
			$y_2=y_1-(y_1-y_0)\tau$, $y_3=y_0+(y_1-y_0)\tau$\;
		}
		$t_1=\lfloor\log_{p_2} y_0\rfloor$, $t_2=\lceil\log_{p_2} y_1\rceil$\;
		$m=\argmin_{m\in\{t_1,t_2\}} J_1({p_2}^{m};k,\beta)$\;
		}		
	 \If{$J_1(p_2^{m}, k_1,\beta)\leq J^*$}{
	 $J^*=J_1(p_2^{m}, k_1,\beta)$\;
	}	
	}
	}	
	\eIf{$J^*\geq0$}{
 			$l=\beta$\;}{
 			$r=\beta$\;
		}
		
}
\end{algorithm}

\textbf{Step (i): Find the optimal policy $\pi^*_\beta$:}
Note that both of the threshold-type policies defined in \eqref{type2} and \eqref{type1} result in a Markov chain with a single positive recurrent class. Thus, given threshold $a_1^*$ in \eqref{type2} or $a_2^*$ in \eqref{type1}, we can obtain the expression of average cost under the corresponding threshold-type policy with aid of state transition diagram. With this, we obtain some nice properties in Theorem \ref{CONVEXITY}, which enables us to get a low-complexity algorithm. Before providing the result, we define the integer threshold which will be used in the theorem and algorithm.

Recall that age $a\in \mathcal{S}$ is expressed as the sum of multiple $d_1$'s and $d_2$'s. 
Note that under the threshold-type policy in \eqref{type2}, if $ a< a_1^*$, $\zeta(a)=2$. This implies that if $ a< a_1^*$, $a$ is in the form $a=d_j+l d_2, j\in\mathcal{E}, l\in \mathbb{N}$. Thus, it is sufficient to use the following integer threshold to represent the threshold-type policy in \eqref{type2}.
\begin{align}
	m_1&=\min \left\{l: d_1+ld_2\geq a_1^*,l \in  \mathbb{N}\right\},\label{3}\\
	n_1&=\min \left\{l: d_2+ld_2\geq a_1^*,l \in  \mathbb{N} \right\}\label{4}.
\end{align}
Similarly, the policy in \eqref{type1} can be represented by
\begin{align}
m_2&=\min \left\{l: d_1+ld_1\geq a_2^*, l \in \mathbb{N}\right\},\label{1}\\
	n_2&=\min \left\{l: d_2+ld_1\geq a_2^*, l \in \mathbb{N} \right\},\label{2}.
\end{align}

Based on the analysis, \eqref{pBeta2} and \eqref{pBeta1} can be re-expressed as
\begin{numcases}{p(\beta)\triangleq}
\min_{m_1, k_1} J_1(m_1, k_1,\beta), \! & \!if $d_1(1-p_2)\!\!<\! d_2(\!1-p_1\!)$ \label{target2}\\
\min_{m_2, k_2} J_2(m_2, k_2,\beta),\! & \!if $d_1(1-p_2)\!\!\geq \!d_2(\!1-p_1\!)$ \label{target1}
\end{numcases}
where $k_1=n_1-m_1$ and $k_2=n_2-m_2$. Also, we use $J_1(m_1, k_1,\beta)$ and $J_2(m_2, k_2,\beta)$ to denote the average cost under the policies in \eqref{type2} and \eqref{type1}, respectively. We have $k_2\in \mathcal{K}_2\triangleq\{0,1\}$ and $k_1\in \mathcal{K}_1\triangleq\big\{0,\cdots,\lfloor\frac{d_1}{d_2}\rfloor\big\}$. In particular, according to the definition of $m_2$ and $n_2$, we have $(m_2+1)d_1\geq a_2^*>d_2+(n_2-1)d_1$ and $d_2+n_2d_1\geq a_2^*>m_2d_1$. Substitute $k_2=n_2-m_2$ into these two inequalities, we get $k_2\in \mathcal{K}_2$. Similarly, we have $k_1\in \mathcal{K}_1$.

\begin{table}[]
\renewcommand\arraystretch{2}
\caption{Notations}
\footnotesize
\begin{tabular}{p{0.8mm}l}
\hline\hline
$A_1$&$= p_2^{k_1}\big(d_2(1-p_1)-d_1(1-p_2)\big)\big(d_2(k_1+1)-d_1\big)$ \\
\hline
$B_1$&$= -A_2+p_2^{k_1}(1-p_2)(0.5d_1^2-\beta d_1+\frac{d_1^2}{1-p_1})+(1-p_1)(\beta d_2-0.5d_2^2-\frac{d_2^2}{1-p_2})$ \\
\hline
$C_1$&$= (1-p_1)(0.5d_2^2-\beta d_2+\frac{1}{1-p_2})d_2^2$ \\
\hline
$D_1$&$= \big(d_1(1-p_2)-d_2(1-p_1)\big)d_2p_2^{k_1}$\\
\hline
$A_2$&$= p_1^{k_2}\big(d_1(1-p_2)-d_2(1-p_1)\big)\big(d_1(1-{k_2})-d_2\big)$\\
\hline
$B_2$&$= -p_1^{k_2}(1-p_2)(0.5d_1^2-\beta d_1+\frac{d_1^2}{1-p_1})+(1-p_1)(0.5d_2^2-\beta d_2+\frac{d_2^2}{1-p_2})$\\
	&$  \ \ \  \
	-\big(d_1(1-p_2)-d_2(1-p_1)\big)(d_1-d_2)$\\
\hline
$C_2$&$= (1-p_2)(0.5d_1^2-\beta d_1+\frac{d_1^2}{1-p_1})$ \\
\hline
$D_2$&$= \big(d_2(1-p_1)-d_1(1-p_2)\big)d_1$ \\
\hline\hline
\end{tabular}
\label{notations}
\end{table}

In Theorem \ref{CONVEXITY}, we provide some nice properties for $J_1$ and $J_2$, which enables us to develop a low-complexity algorithm.
Some notations used in this theorem are defined in Table \ref{notations}.

%
%

\begin{theorem}
\label{CONVEXITY} 
Given $\beta\in [\beta_\text{min}, \beta_\text{max}]$,

(i) if $d_1(1-p_2)<d_2(1-p_1)$, then the average cost under the policy in \eqref{type2} is given by 
\begin{equation}
 	J_1(y, k_1,\beta)=\frac{A_1y^{2}+B_1y+C_1+D_1y\log_{p_2}(y)}{1-p_1+(-1+p_1+p_2^{k_1}(1-p_2))y},
 	\label{cost2}
 \end{equation}
where $y\triangleq p_2^{m_1}$. Moreover, $J_1(y, k_1,\beta)$ is quasi-convex in $y$ for $0<y\leq 1$, given $k_1\in\mathcal{K}_1$.

(ii) if $d_1(1-p_2)\geq d_2(1-p_1)$, then optimal average cost is given by
	\begin{equation}
		J_2(0, 0,\beta)=\frac{A_2+B_2+C_2}{1-p_1}.
		\label{cost1}
	\end{equation}
Moreover, the optimal policy for \eqref{target1} chooses $u=2$ at every transmission opportunity.
\end{theorem}
\ifreport
\begin{proof}
See Appendix \ref{convex_prop}.
\end{proof}
\fi
 
With the property in Theorem \ref{CONVEXITY} (i), we are able to use golden section search \cite{press2007section} to find the optimal value of $y$ under condition $d_1(1-p_2)<d_2(1-p_1)$. The details are provided in Algorithm \ref{alg1} (Line 12-20).
In Theorem \ref{CONVEXITY} (i), we make a change of variable, i.e., $m_1$ is replaced with $\log_{p_2}(y)$ in $J_1$, where $y\in (0,1]$. Note that $\log_{p_2}(\cdot)$ is one-to-one functions. Thus, after obtaining the optimal $y$ that minimizes $J_1$ , the corresponding optimal threshold $m_1$ can be easily obtained by comparing 
$J_1({p_2}^{\lfloor \log_{p_2}(y)\rfloor}, k_1,\beta)$ and $J_1({p_2}^{\lceil \log_{p_2}(y)\rceil}; k_1,\beta)$). The details are provided in Algorithm \ref{alg1} (Line 21-22).
Till now, we have solved $\min_{m_1} J_1(m_1, k_1,\beta)$ given $k_1$ and $\beta$. Note that $k_1\in\mathcal{K}_1$ has finite and countable values. Then, we can easily solve \eqref{target2} under condition $d_1(1-p_2)< d_2(1-p_1)$ by searching for the optimal $k_1$ in the finite set $\mathcal{K}_1$.

Note that Theorem \ref{CONVEXITY} (ii) applies to all $\beta\in \![\beta_\text{min},\! \beta_\text{max}\!]$ including $\beta^*$. Thus, the optimal policy for \eqref{type1} is to take $u=2$ at every transmission opportunity, which is returned directly in Algorithm \ref{alg1} (Line 4-5). Under this condition, the step (ii) is only used to find optimal average cost $\beta^*$. 

Moreover, in Line 9-10, we add a judgement sentence, i.e., if the condition $\frac{\partial J_1(y,k_1,\beta)}{\partial y}\vert_{y=1}<0$ is satisfied, then we can directly obtain the optimal $m_1$ without running golden section method. This further reduces the algorithm complexity. The judgement is based on the fact that $\lim_{y\rightarrow 0}\frac{\partial J_1(y, k_1,\beta)}{\partial y}<0$ (this is proved in the proof of Theorem \ref{CONVEXITY} in Appendix \ref{convex_prop}) and $J_1$ is quasi-convex. Thus, if $\frac{\partial J_1(y, k_1,\beta)}{\partial y}\vert_{y=1}<0$, then $J_1$ is non-increasing in $y$.

\textbf{Step (ii): Find $\beta^*$:} By Lemma \ref{TRANS}, if $p(\beta)>0$, then $\beta<\beta^*$; if $p(\beta)<0$, then $\beta>\beta^*$. Thus, we can use bisection method to search for $\beta^*$. The details are provided in Algorithm \ref{alg1} (Line 2-3 and Line 29-33).

\section{A Discussion on the General Multi-Rate Transmission Selection Problem}
\label{diss}
For the multi-rate transmission selection (from more than two-rate) problem, we assume that there are $N\in\mathbb{N}^+$ transmissions for selection such that transmission delays and error probabilities satisfy $d_j>d_{j+1}$ and $p_j<p_{j+1}$, for $j\in \{1,2,\cdots,N\}$, respectively. Since the transmission delays and transmission error probabilities affect the age in opposite direction, it may be difficult to determine the optimal policy. Thanks to the results obtained for the two-rate transmission selection problem, we obtain some useful insights for the general multi-rate transmission selection. 

In particular, for the two-rate transmission selection problem, the optimal decision is monotonically non-increasing in age under condition $\bar{d}_1\leq \bar{d}_2$ and monotonically non-decreasing in age under condition $\bar{d}_1\geq \bar{d}_2$, where $\bar{d}_1$ and $\bar{d}_2$ are mean delays of low and high rate transmissions, respectively, as defined in \eqref{mean_delay}. The monotonic property is obtained by showing that Q function $Q^\alpha(a,u)$ in \eqref{qvalue} is supermodular (submodular) in $(a,u)$ under former (later) condition. If we apply the same technique to the multi-rate transmission selection, we will be able to get the following properties:
\begin{itemize}
	\item (i) if $\bar{d}_1\leq \bar{d}_2\leq\cdots\leq\bar{d}_N$, then the optimal decision will be monotonically non-increasing in age,
	\item (ii) if $\bar{d}_1\geq \bar{d}_2\geq\cdots\geq\bar{d}_N$, then the optimal decision will be monotonically non-decreasing in age.
\end{itemize}
Specifically, for (i), to show the monotonic property, we can show that given transmission $l \in\{1,2,\cdots,N\}$, any transmission $v$ with longer delay, i.e. $d_v>d_l$, satisfy that Q function is supermodular in $(a, u)$. Since the state space is not changed, it will be provable with this technique. The same analysis applies to (ii).
However, for the cases that are not covered in (i) and (ii), we need more investigation on the property of the optimal policy and may use other techniques for the proof, which is for our future work.

%

In addition, if we can show that the optimal policy for the general multi-rate problem is of threshold-type, then machine learning algorithms can be used to determine the optimal threshold. For example, if we regard each threshold-type policy with a certain threshold as a bandit, then basic bandit algorithms like UCB can be exploited to find the optimal threshold \cite{prabuchandran2016reinforcement}. 
 
\section{Numerical Results}
\label{sec5}
In this section, we present some numerical results to explore the performance of the threshold-based age-optimal policy and verify our theoretical results. 

First, we consider an update system, in which $p_1=0.4$ and $p_2=0.75$. 
Table. \ref{thre_delay} illustrates the relation between the optimal threshold versus the transmission delay $d_2$ and the delay ratio $\frac{d_1}{d_2}$ under condition $d_2(1-p_1)>d_1(1-p_2)$. The threshold $(m_1, n_1)$ in the table is obtained by Algorithm \ref{alg1}. 
We observe that the threshold increases with either $d_2$ or the delay ratio $\frac{d_1}{d_2}$. Note that when the age is below the threshold, transmission rate $2$ is selected. Thus, this observation implies that transmission rate $2$ becomes more preferable either when $d_2$ increases with fixed $\frac{d_1}{d_2}$ or when $\frac{d_1}{d_2}$ increases with fixed $d_2$.

\begin{table}[]
\caption{Optimal threshold versus delay}
\begin{tabular}{c|ccccc}
\hline
        & $d_1\!\!=\!\!1.5d_2$ & $d_1\!\!=\!\!1.7d_2$ & $d_1\!\!=\!\!1.9d_2$ & $d_1\!\!=\!\!2.1d_2$ & $d_1\!\!=\!\!2.3d_2$ \\
        \hline
$d_2\!\!=\!\!1$ & (0,0)        & (0,0)        & (1,2)        & (3,4)        & (15,16)      \\
$d_2\!\!=\!\!5$ & (0,1)        & (0,1)        & (1,2)        & (3,4)        & (15,16)      \\
$d_2\!\!=\!\!9$ & (0,1)        & (0,1)        & (1,2)        & (3,4)        & (15,16)     \\
\hline
\end{tabular}
\label{thre_delay}
\end{table}

In Fig. \ref{age_error}, we consider an update system, in which the transmission delays are $d_1=10$ and $d_2=8$. We use ``Delay-Optimal'' to denote the optimal policy that minimizes the average delay by always choosing the transmission rate with minimum mean delay \cite{ozkan2014optimal,de2005managing}. Moreover, we use ``Age-Optimal'' to denote the optimal policy that is obtained from Algorithm \ref{alg1}. We use ``Random $p$'' to denote the policy that chooses transmission rate $1$ with probability $p$. We compare our threshold-based ``Age-Optimal'' policy with ``Random $p$'' policies and the ``Delay-Optimal'' policy, where $p\in \{0.25,0.5\}$.
\begin{figure}[]
 \subfloat[Average age vs $p_1$ given $p_2=0.5$]{
 \includegraphics[scale=.316]{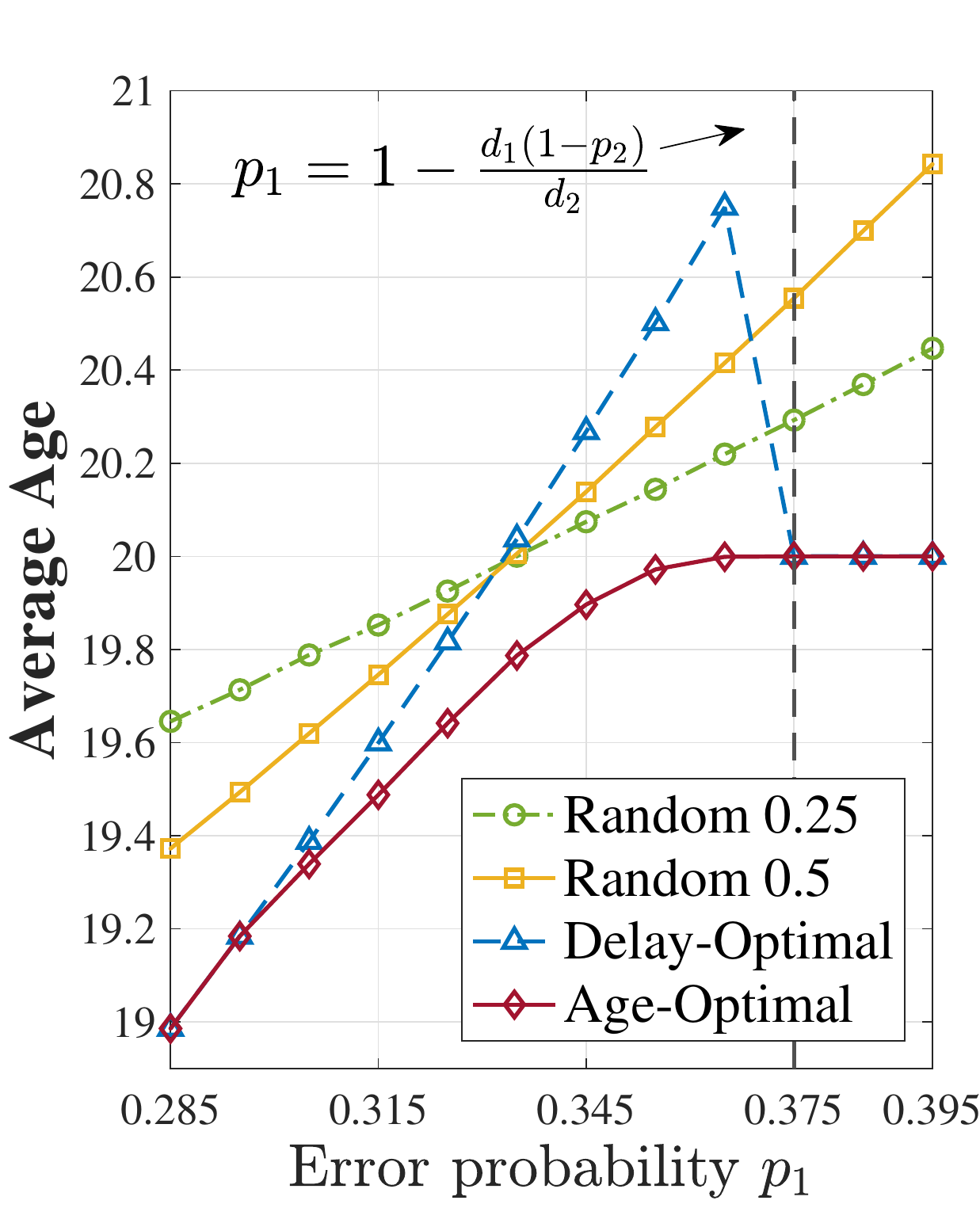}
 \label{age_p1}
 }\ \ \
 \subfloat[Average age vs $p_2$ given $p_1=0.5$]{
 \includegraphics[scale=.316]{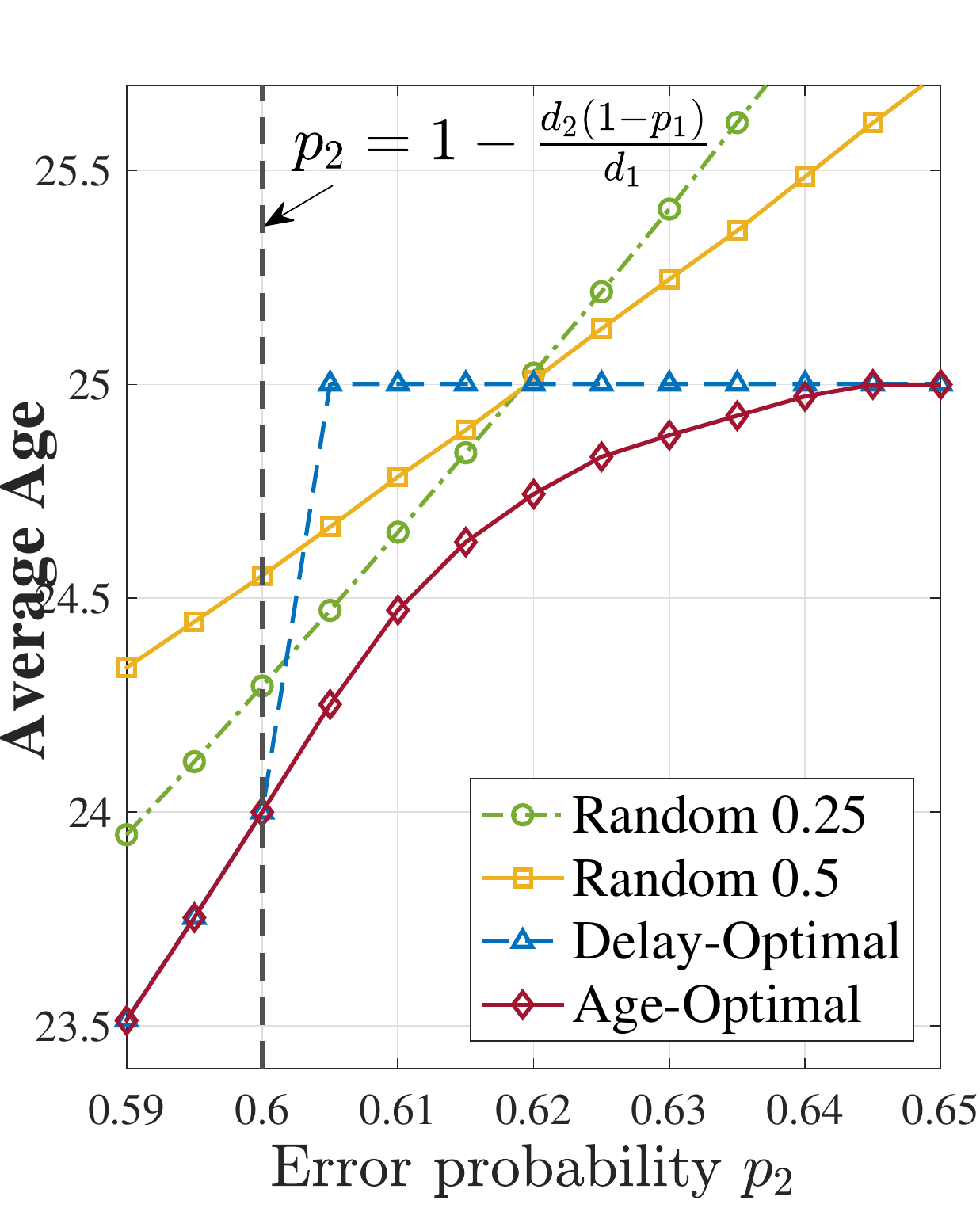}
 \label{age_p2}
 }
 \caption{Total average age versus transmission error probability}
	\label{age_error}
\vspace{-0.2cm}
\end{figure}
Fig. \ref{age_p1} (\ref{age_p2}) illustrates the total average age in \eqref{total_avg_age} versus transmission error probability $p_1$ ($p_2$) given $p_2=0.5$ ($p_1=0.5$). The dashed line in the figure marks the point at which $d_1(1-p_2)=d_2(1-p_1)$. The left and right (right and left) of the line corresponds to conditions $d_1(1-p_2)<d_2(1-p_1)$ and $d_1(1-p_2)>d_2(1-p_1)$ in Fig. \ref{age_p1} (in Fig. \ref{age_p2}), respectively. As we can observe, the age-optimal policy outperforms other plotted policies. This agrees with Theorem \ref{THRE_PROP}. Moreover, the results confirm that the delay-optimal policy does not necessarily minimize the age. In particular, the gap between delay-optimal and age-optimal policy becomes larger as $p_1$ ($p_2$) approaches to the left side (right side) of the dashed line in Fig. \ref{age_p1} (Fig. \ref{age_p2}). The jump in the curve of the delay-optimal policy is incurred by the switch between two transmission rates. For example, in Fig. \ref{age_p1}, delay-optimal policy chooses $u=1$ on the left side of the dashed line, while chooses $u=2$ on the right side. This is because transmission rate $1$ has smaller mean delay on the left side while transmission rate $2$ has smaller mean delay on the right side.

\section{Conclusion}
In this paper, we studied the transmission selection problem for minimizing age of information in information update system with heterogenous transmissions. We assume that there are two different transmissions with varying delay and error probability. We showed that there exists a stationary deterministic optimal transmission selection policy which is of threshold-type in age (Theorem \ref{THRE_PROP}). This result reveals an interesting phenomenon: If the mean delay of the low rate transmission is smaller than that of high rate transmission, then the optimal action chooses the one with higher mean delay when age is smaller than a certain threshold. This is in contrary with the delay-optimal policy that always chooses the transmission with lower mean delay. In addition, we showed that if the mean delay of the low rate transmission is smaller than that of high rate transmission, the average cost is quasi-convex in a threshold related variable; otherwise, the optimal policy chooses $u=2$ for each transmission opportunity (Theorem \ref{CONVEXITY}). This enabled us to design a low-complexity algorithm to obtain the optimal policy (Algorithm \ref{alg1}). For the future work, we plan to study the multi-rate scenario with more than two selections of heterogenous transmissions based on the insights discussed in Section \ref{diss}.

\bibliographystyle{unsrt}
\bibliography{Reference}
\appendix
\section{Proof of Lemma \ref{TRANS}}
\label{appe0}
We show (i) with two steps.

\textit{Step1:} This step proves that $\bar{\Delta}^* \leq \beta$ if and only if $p(\beta) \leq 0$.

If $\bar{\Delta}^* \leq \beta$, there exists a $\pi$ that is feasible for \eqref{equiv1} and \eqref{equiv2} which satisfies 
\begin{align}
	\limsup_{n\rightarrow \infty}	\frac{\sum_{0}^{n-1} \mathbb{E}[a_i Y_{i+1}+\frac{1}{2}Y_{i+1}^2]}{\sum_{i=0}^{n-1}\mathbb{E}[ Y_{i+1}]}\leq \beta \label{l1_1}.
\end{align}
Therefore, we have 
\begin{align}
	\limsup_{n\rightarrow \infty}	\frac{\frac{1}{n}\sum_{0}^{n-1} \mathbb{E}[(a_i-\beta) Y_{i+1}+\frac{1}{2}Y_{i+1}^2]}{\frac{1}{n}\sum_{i=0}^{n-1}\mathbb{E}[ Y_{i+1}]}\leq 0 \label{l1_2}.
\end{align}
Since $Y_i \in \mathcal{D}$, $Y_i$'s are bounded and positive and $\mathbb{E}[Y_i]>0$ for all $i$. Then, we have $$0< \liminf _{n \rightarrow \infty}\frac{1}{n}\sum_{i=0}^{n-1}\mathbb{E}[ Y_{i+1}]\leq \limsup_ {n \rightarrow \infty}\frac{1}{n}\sum_{i=0}^{n-1}\mathbb{E}[ Y_{i+1}]\leq \max \{d_1,d_2\}.$$ 
Thus, we obtain
\begin{align}
	\limsup_{n\rightarrow \infty}\frac{1}{n}\sum_{0}^{n-1} \mathbb{E}[(a_i-\beta) Y_{i+1}+\frac{1}{2}Y_{i+1}^2]\leq 0 \label{l1_3}.
\end{align}
That is, $p(\beta)\leq 0$.
For the reverse direction, if $p(\beta)\leq 0$, then there exists a $\pi$ that is feasible for \eqref{equiv1} and \eqref{equiv2} which satisfies \eqref{l1_3}. Since $$0< \liminf _{n \rightarrow \infty}\frac{1}{n}\sum_{i=0}^{n-1}\mathbb{E}[ Y_{i+1}]\\ \leq \limsup_ {n \rightarrow \infty}\frac{1}{n}\sum_{i=0}^{n-1}\mathbb{E}[ Y_{i+1}]\leq \max \{d_1,d_2\},$$ we can divide \eqref{l1_3} by $\frac{1}{n}\sum_{i=0}^{n-1}\mathbb{E}[ Y_{i+1}]$ to get \eqref{l1_2}, which implies \eqref{l1_1}. Thus, we have $\bar{\Delta}^* \leq \beta$.

\textit{Step 2:} We can prove that $\bar{\Delta}^* < \beta$ if and only if $p(\beta) < 0$ with same argument in Step 1, where $\leq$ is replaced with $<$. Finally, from the Step 1, it immediately follows that $\bar{\Delta}^* > \beta$ if and only if $p(\beta) > 0$

Part (ii): By (i), $p(\beta)=0$ is equivalent to $\bar{\Delta}^*=\beta$. With this condition, we firstly show that each optimal solution to \eqref{equiv1} is an optimal solution to \eqref{equiv2}. Suppose policy $\pi$ is an optimal solution to \eqref{equiv1}. Then, $\bar{\Delta}(\pi)=\bar{\Delta}^*=\beta$. With similar argument of \eqref{l1_1}-\eqref{l1_3}, we can get that policy $\pi$ satisfies
\begin{align}
	\limsup_{n\rightarrow \infty}	\frac{1}{n}\sum_{0}^{n-1} \mathbb{E}[(a_i-\beta) Y_{i+1}+\frac{1}{2}Y_{i+1}^2]= 0 .
\end{align}
Since $p(\beta)=0$ (optimal value for \eqref{equiv2} is 0), policy $\pi$ is an optimal solution to \eqref{equiv2}.

Similarly, we can prove that each optimal solution to \eqref{equiv2} is an optimal solution to \eqref{equiv1}.

\section{Proof of Proposition \ref{DISC_EXIST}}
\label{proposition}
Let $w$ be a positive real-valued function on $\mathcal{S}$ defined by 
\begin{align}
	w(a)=\max(ad_1+0.5d_1^2,1).
\end{align}
Define the weighted supremum norm $\vert\vert\cdot\vert\vert_w$ for real-valued functions $f$ on $\mathcal{S}$ by
$$ \vert\vert f\vert\vert_w=\sup_{a\in \mathcal{S}}\frac{f(a)}{w(a)},$$
and let $\mathcal{F}$ be the space of real-valued functions $f$ on $\mathcal{S}$ that satisfies $\vert\vert f\vert\vert_w<\infty$. Note that $V_0^\alpha \in \mathcal{F}$. By Theorem 6.10.4 in \cite{puterman2014markov}, we only need to show the following two assumptions hold:

\noindent\emph{Assumption 1:} There exists a constant $\mu<\infty$ such that 
\begin{align}
	\max_{u\in \mathcal{U}}\vert C(a,u)\vert\leq\mu w(a).
	\label{assum1}
\end{align}
\emph{Assumption 2:}
(i) There exists a constant $L$, $0\leq L<\infty$, such that
\begin{align}
	\sum_{a^\prime\in\mathcal{S}}P(a^\prime|a,u)w(a^\prime)\leq w(a)L,
	\label{assum2i}
\end{align}
for all $u$ and $a$.

(ii) For each $\alpha$, there exists an $\gamma$, $0\leq\gamma<1$ and an integer $M$ such that for all deterministic policies $\pi$,
\begin{align}
\alpha^M \sum_{a^\prime\in\mathcal{S}}P_\pi^M(a^\prime|a)w(a^\prime)\leq \gamma w(a)	,
\label{assum2ii}
\end{align}
where $P_\pi^M$ denotes the $M$-stage transition probability under policy $\pi$.

First, we show that Assumption 1 holds with $\mu=\max\{d_1\beta,1\}$. Recall that $C(a,u)=(a-\beta)d_u+0.5d_u^2$. If $(a-\beta)d_u+0.5d_u^2\geq 0$, then $\vert C(a,u)\vert=(a-\beta)d_u+0.5d_u^2\leq ad_u+0.5d_u^2\leq ad_1+0.5d_1^2\leq \mu w(a)$; otherwise, $\vert C(a,u)\vert=\beta d_u-ad_u-0.5d_u^2\leq\beta d_u\leq \beta d_1\leq \mu w(a)$.
Next, we show that Assumption 2 holds. Before providing the proof for Assumption 2, we first show that for all $l\in\mathbb{N}^+$, \eqref{tmp} holds, which will be used in our proof for Assumption 2. 
\begin{align}
	w(a+ld_1)\leq (2l+1+ld_1^2)w(a).
	\label{tmp}
\end{align}
We show that \eqref{tmp} holds by analyzing different cases:

\emph{Case 1:} If $(a+ld_1)d_1+0.5d_1^2\leq 1$, then $w(a+ld_1)=w(a)=1$. Thus, $w(a+ld_1)\leq (2l+1+ld_1^2)w(a)$.

\emph{Case 2:} If $(a+ld_1)d_1+0.5d_1^2>1$ and $ad_1+0.5d_1^2>1$, then we have
\begin{align}
	&w(a+ld_1)-(2l+1+ld_1^2)w(a)\notag\\
	=&(a+ld_1)d_1+0.5d_1^2-(2l+1+ld_1^2)(ad_1+0.5d_1^2)\\
	=&-2lad_1-lad_1^3-0.5ld_1^4\\
	<& 0.
\end{align}

\emph{Case 3:} If $(a+ld_1)d_1+0.5d_1^2>1$ and $ad_1+0.5d_1^2<1$, then $ad_1<1-0.5d_1^2$ and we have
\begin{align}
	w(a+ld_1)&=(a+ld_1)d_1+0.5d_1^2\\
	&\leq 1-0.5d_1^2+ld_1^2+0.5d_1^2\\
	&\leq 2l+1+ld_1^2\\
	&=(2l+1+ld_1^2)w(a).
\end{align}
This completes the proof of \eqref{tmp}. Next, using \eqref{tmp}, we show Assumption 2 holds.
Note that $w(a)$ is an non-decreasing function with $a$. With this, we have
\begin{align}
	\sum_{a^\prime\in\mathcal{S}}P(a^\prime|a, u)w(a^\prime)&=p_u w(a+d_u)+(1-p_u) w(d_u)\\
	&\leq w(a+d_1)\\
	&\leq (3+d_1^2)w(a),
\end{align}
where the last inequality holds by \eqref{tmp}.
Hence, Assumption 2(i) holds with $L=3+d_1^2$. Similarly, for any deterministic policy $\pi$, we have 
\begin{align}
	\alpha^M \sum_{a^\prime\in\mathcal{S}}P_\pi^M(a^\prime|a)w(a^\prime)&\leq \alpha^M w(a+Md_1)\\
	&\leq \alpha^M (2M+1+Md_1^2)w(a).
\end{align}
Consequently, for $M$ sufficiently large, $\alpha^M (2M+1+Md_1^2)<1$. Hence, Assumption 2(ii) holds.

\section{Proof of Lemma \ref{THRESHOLD}}
\label{discount_threshold}
	(i) Let $a_1$ and $a_2$ be the age such that $a_1\geq a_2$. We want to show that if $Q^\alpha(a_2,1)\leq Q^\alpha(a_2,2)$, then $Q^\alpha(a_1,1)\leq Q^\alpha(a_1,2)$. It suffices to show that $Q^\alpha(a,1)-Q^\alpha(a,2)$ is decreasing with age $a$, i.e.,
	\begin{align}
		Q^\alpha(a_2,1)-Q^\alpha(a_2,2)\geq Q^\alpha(a_1,1)-Q^\alpha(a_1,2).
	\end{align}
	By Proposition \ref{DISC_EXIST}, we only need to show that for $n\in \mathbb{N}$, 
	\begin{align}
		&Q^\alpha_{n+1}(a_2,1)-Q^\alpha_{n+1}(a_2,2)\geq Q^\alpha_{n+1}(a_1,1)-Q^\alpha_{n+1}(a_1,2)\\
		\Leftrightarrow &C(a_2,1)+\alpha p_1V_n^\alpha(a_2+d_1)+\alpha(1-p_1)V_n^\alpha(d_1)\notag\\
		&-C(a_2,2)-\alpha p_2V_n^\alpha(a_2+d_2)-\alpha(1-p_2)V_n^\alpha(d_2)\notag\\
		&\geq 
		C(a_1,1)+\alpha p_1V_n^\alpha(a_1+d_1)+\alpha(1-p_1)V_n^\alpha(d_1)\notag\\
		&-C(a_1,2)-\alpha p_2V_n^\alpha(a_1+d_2)-\alpha(1-p_2)V_n^\alpha(d_2)\\
		\Leftrightarrow &(a_2-a_1)(d_1-d_2)-\alpha p_1V^\alpha_{n}(a_1+d_1)+\alpha p_1V^\alpha_{n}(a_2+d_1)\notag\\
		&-\alpha p_2V^\alpha_{n}(a_2+d_2)+\alpha p_2V^\alpha_{n}(a_1+d_2)\geq0.
		\label{st1}
	\end{align}
	We show \eqref{st1} by induction. When $n=0$, substitute $V_0^\alpha(a)=\frac{d_1-d_2}{\alpha(p_2-p_1)}a$ into the left-hand-side of \eqref{st1} with $n$ replaced with 0 and we have 
	\begin{align}
		&(a_2-a_1)(d_1-d_2)-\alpha p_1V^\alpha_{0}(a_1+d_1)+\alpha p_1V^\alpha_{0}(a_2+d_1)\notag\\
		&-\alpha p_2V^\alpha_{0}(a_2+d_2)+\alpha p_2V^\alpha_{0}(a_1+d_2)\notag\\
		&=0
	\end{align}
Hence, \eqref{st1} holds when $n=0$. Suppose \eqref{st1} holds for $n$, we will show that it holds for $n+1$.
Let $u_1, u_2, u_3, u_4$ be the optimal actions in state $a_1+d_1, a_2+d_1, a_2+d_2, a_1+d_2$, respectively. Specifically, $V^\alpha_{n+1}(a_1+d_1)=Q^\alpha_{n+1}(a_1+d_1,u_1)$, $V^\alpha_{n+1}(a_2+d_1)=Q^\alpha_{n+1}(a_2+d_1,u_2)$, $V^\alpha_{n+1}(a_2+d_2)=Q^\alpha_{n+1}(a_2+d_2,u_3)$ and $V^\alpha_{n+1}(a_1+d_2)=Q^\alpha_{n+1}(a_1+d_2,u_4)$. Then, the left-hand-side of \eqref{st1} is 
\begin{align}
	&(a_2-a_1)(d_1-d_2)\notag\\
	&-\alpha p_1Q^\alpha_{n+1}(a_1+d_1,u_1)+\alpha p_1Q^\alpha_{n+1}(a_2+d_1,u_2)\notag\\
		&-\alpha p_2Q^\alpha_{n+1}(a_2+d_2,u_3)+\alpha p_2Q^\alpha_{n+1}(a_1+d_2,u_4)\notag\\
=&(a_2-a_1)(d_1-d_2)\notag\\
&+\underbrace{\alpha p_1Q^\alpha_{n+1}(a_2+d_1,u_2)-\alpha p_1Q^\alpha_{n+1}(a_1+d_1,u_2)}_{A}\notag\\
&+\underbrace{\alpha p_1Q^\alpha_{n+1}(a_1+d_1,u_2)-\alpha p_1Q^\alpha_{n+1}(a_1+d_1,u_1)}_{\geq 0 \ (\text{By optimality of action } u_1)}\notag\\
&-(\underbrace{\alpha p_2Q^\alpha_{n+1}(a_2+d_2,u_4)-\alpha p_2Q^\alpha_{n+1}(a_1+d_2,u_4)}_{B})\notag\\
&+\underbrace{\alpha p_2Q^\alpha_{n+1}(a_2+d_2,u_4)-\alpha p_2Q^\alpha_{n+1}(a_2+d_2,u_3)}_{\geq 0 \ (\text{By optimality of action } u_3)}.
\label{t0}
\end{align}
By induction hypothesis, we have
\begin{align}
	A=&\alpha p_1 \min_{u_2\in \{1,2\}}\big\{(a_2-a_1)d_{u_2}+p_{u_2}V_n^\alpha(a_2+d_1+d_{u_2})\notag\\
	&-p_{u_2}V_n^\alpha(a_1+d_1+d_{u_2}) \big\}\\
	\geq &\alpha p_1\Big((a_2-a_1)d_2+p_{2}V_n^\alpha(a_2+d_1+d_2)\notag\\
	&-p_{2}V_n^\alpha(a_1+d_1+d_2)\Big)
	\label{t2}.
\end{align}
Similarly, 
\begin{align}
	B\leq &\alpha p_2\Big((a_2-a_1)d_1+p_{1}V_n^\alpha(a_2+d_2+d_1)\notag\\
	&-p_{1}V_n^\alpha(a_1+d_2+d_1)\Big).
	\label{t3}
\end{align}
Hence, substitute \eqref{t2} and \eqref{t3} into \eqref{t0} and we obtain
\begin{align}
&(a_2-a_1)(d_1-d_2)\notag\\
	&-\alpha p_1Q^\alpha_{n+1}(a_1+d_1,u_1)+\alpha p_1Q^\alpha_{n+1}(a_2+d_1,u_2)\notag\\
		&-\alpha p_2Q^\alpha_{n+1}(a_2+d_2,u_3)+\alpha p_2Q^\alpha_{n+1}(a_1+d_2,u_4)\notag\\
	\geq & (a_2-a_1)(d_1-d_2)+A-B\\
	\geq &(a_2-a_1)\left((1-\alpha p_2)d_1-(1-\alpha p_1)d_2\right)\\
	\geq & 0,\label{num1}
\end{align}
where \eqref{num1} holds by the condition $(1-\alpha p_2)d_1\leq(1-\alpha p_1)d_2$.

(ii) With similar analysis in part (i), it suffices to show that for $a_1\geq a_2$,
\begin{align}
	Q^\alpha(a_2,1)-Q^\alpha(a_2,2)\leq Q^\alpha(a_1,1)-Q^\alpha(a_1,2).
\end{align}
The proof is similar to part (i). 
\section{Proof of Lemma \ref{CANDI}}
\label{candidate_policy}
Part (i) is a direct result of the Lemma in \cite{sennott1989average}. For (ii), we have two cases:

Case 1:
If $d_1(1-p_2)\leq d_2(1-p_1)$ (that is, $\frac{d_1-d_2}{d_1p_2-d_2p_1}<1$), then there exists an integer $M$ such that $\gamma_n>\frac{d_1-d_2}{d_1p_2-d_2p_1}$ (that is, $d_1(1-\gamma_n p_2)<d_2(1-\gamma_n p_1)$) for all $n\geq M$. Thus, for all $n\geq M$, $\pi^{\gamma_n}$ is of the threshold-type that is defined in \eqref{type2}. This implies that $\pi^\star$ is of the threshold-type defined in \eqref{type2}. 

Case 2: If $d_1(1-p_2)\geq d_2(1-p_1)$ (that is, $\frac{d_1-d_2}{d_1p_2-d_2p_1}\geq 1$), then for all $n$, $\gamma_n<1\leq\frac{d_1-d_2}{d_1p_2-d_2p_1}$ (that is, $d_1(1-\gamma_n p_2)>d_2(1-\gamma_n p_1)$). In this case, $\pi^{\gamma_n}$ is of threshold-type defined in \eqref{type1}. Then, $\pi^\star$ is of threshold-type defined in \eqref{type1}.

\section{Proof of Condition Verification in Theorem \ref{THRE_PROP}}
\label{app1}

By \cite{sennott1989average}, we need to verify the following two conditions, A1 and A2, hold.
\begin{itemize}
	\item A1: There exists a deterministic stationary transmission selection policy $\pi'$ such that the resulting Markov chain is irreducible and aperiodic, and the average cost $J(\pi',\beta)$ is finite.
	\item A2: There exists a non-negative $L$ such that $-L\leq h^\alpha (a)\triangleq V^\alpha(a)-V^\alpha(a_0)$ for all $a$ and $\alpha$, where $a_0$ is a reference state.
\end{itemize}

For A1, let $\pi'$ be the policy that always chooses $u=2$. It is obvious that the resulting Markov chain is irreducible and aperiodic. 
By \cite{hsu2019scheduling}, if we regard dynamics of the age as a queue system, then the average arrival rate is $d_2$ per stage since age increases by $d_2$ per stage, and the average service rate is infinite since each successful delivery reduces age (can be very large) to $d_2$ and successful delivery occurs with positive probability $1-p_2$. Thus, average age-queue size is finite, i.e.,  $\limsup_{n \rightarrow \infty} \frac{1}{n}\mathbb{E_{\pi'}}\sum_{i=0}^{n-1} a_i<\infty$. Thus, 
\begin{align}
	J(\pi',\beta)=\limsup_{n \rightarrow \infty} \frac{1}{n}\mathbb{E_{\pi'}}\left[\sum_{i=0}^{n-1} ((a_i-\beta)d_2+0.5d_2^2)\right]<\infty.
\end{align}

For A2, if we are able to show that $V^\alpha(a)$ is non-decreasing in age $a$, then $h^\alpha (a)$ is non-decreasing in age $a$. This implies that $h^\alpha (a)\geq h^\alpha (0)=V^\alpha(0)-V^\alpha(a_0)$, since $0$ is the smallest age in $\mathcal{S}$. Then, A2 holds with $L=-V^\alpha(0)+V^\alpha(a_0)\geq 0$. Now we prove that $V^\alpha(a)$ is non-decreasing in age $a$.

\begin{lemma}
\label{mono}
	Given discount factor $\alpha$, the value function $V^\alpha(a)$ is non-decreasing with age $a$.
\end{lemma}
\begin{proof}
We use induction to show that for all $n\geq 0$, $V_n^\alpha(a_1)\geq V_n^\alpha(a_2)$ holds whenever $a_1\geq a_2$.
	Then, since $V_{n}^\alpha(a)\rightarrow V^\alpha(a)$ as $n\rightarrow \infty$ by Proposition \ref{DISC_EXIST}, we can conclude that $V^\alpha(a)$ is non-decreasing with age $a$.	 
	
	When $n=0$, since $\frac{d_1-d_2}{\alpha(p_2-p_1)}\geq 0$, $V_0^\alpha(a)=\frac{p_1-d_2}{\alpha(p_2-p_1)}a$ is non-decreasing with age $a$. It remains to show that given that $V_n^\alpha(a_1)\geq V_n^\alpha(a_2)$, the inequality $V_{n+1}^\alpha(a_1)\geq V_{n+1}^\alpha(a_2)$ holds. In particular, we have 
	\begin{align}
		Q^\alpha_{n+1}(a_1,u)=&C(a_1,u)\!+\!\alpha p_uV_n^\alpha(a_1+d_u)\!+\!\alpha(1-p_u)V_n^\alpha(d_u)\\
		\geq &C(a_2,u)\!+\!\alpha p_uV_n^\alpha(a_2+d_u)\!+\!\alpha(1-p_u)V_n^\alpha(d_u)\label{tmp2}\\
		=&Q^\alpha_{n+1}(a_2,u),
	\end{align}
	where $\eqref{tmp2}$ holds by $C(a_1,u)>C(a_2,u)$ and induction hypothesis.
	Hence, we have
	\begin{align}
		V_{n+1}^\alpha(a_1)=&\min_{u\in\{1,2\}} Q^\alpha_{n+1}(a_1,u)
		\\
		\geq &\min_{u\in\{1,2\}} Q^\alpha_{n+1}(a_2, u)\\
		=& V_{n+1}^\alpha(a_2).
	\end{align}	
\end{proof}

\section{Proof of Lemma \ref{BETA_RANGE}}
\label{range}
We construct an infeasible policy that chooses $u=2$ at every transmission opportunity, where we assume that the transmission is error-free. 
In this case, we obtain a lower bound $\beta_\text{min}$ of $\beta$ given by
\begin{align}
	\beta_\text{min}=\limsup_{n\rightarrow \infty} \frac{\sum_{0}^{n-1}(d_2^2+\frac{1}{2}d_2^2)}{\sum_{0}^{n-1} d_2}=\frac{3}{2}d_2.
\end{align}

We use $\pi_1$ and $\pi_2$ to denote the policies that always choose $u=1$ and $u=2$, respectively. Moreover, we use $\bar{\Delta}_1$ and $\bar{\Delta}_2$ to denote the average expected age under policy $\pi_1$ and $\pi_2$, respectively. By optimality, $\beta^*=\bar{\Delta}^*\leq \min\{\bar{\Delta}_1,\bar{\Delta}_2\}$. 

Next, we calculate $\bar{\Delta}_1$.
Note that $\pi_1$ results in an Markov chain with a single recurrent class. The state transition diagram is given in Fig. \ref{trans_diag}.
\begin{figure}
	\centering
	\includegraphics[width=0.38\textwidth]{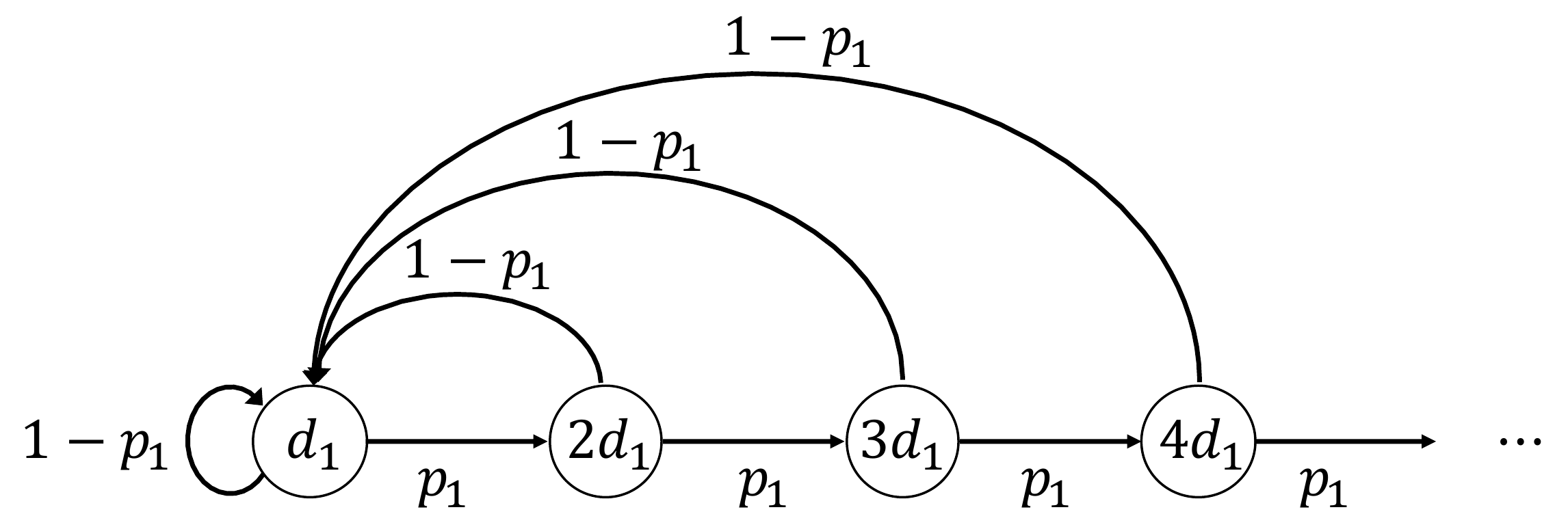}
	\caption{State transition diagram under policy $\pi_1$ that always choose $u=1$}
	\label{trans_diag}
\end{figure}
Let $x_l$ be the steady state probability of the state $a=ld_1$, $l\in \mathbb{N}^+$ when policy $\pi_1$ is used.
Based on the state transition diagram, balance equations can be obtained as follows:
\begin{align}
	& p_1x_{l}=x_{l+1},	\ \ \ \ \forall l\in \{1,2\cdots\} .
\end{align}
Then, $x_l$ is expressed as 
\begin{align}
	x_{l}=x_1p_1^{l-1},	\ \ \ \    \forall l\in \{1,2\cdots\}.
\end{align}
Since $\sum_{l=1}^{\infty}x_l=1$, we obtain $x_1=1-p_1$.
Then, under policy $\pi_1$, the expected age $\mathbb{E}_{\pi_1} [a]$ is given by
\begin{align}
	\mathbb{E}_{\pi_1} [a]=\sum_{l=1}^\infty ld_1x_l=\sum_{i=1}^\infty lp_1^{l-1}(1-p_1)d_1=\frac{1}{1-p_1}d_1.
\end{align}
Under policy $\pi_1$, $Y_{i}=d_1$ for all $i$. Then, $\bar{\Delta}_1$ is given by
\begin{align}
	\bar{\Delta}_1&=\limsup_{n\rightarrow \infty}	\frac{\sum_{0}^{n-1} \mathbb{E}_{\pi_1}[a_i d_{1}+\frac{1}{2}d_{1}^2]}{\sum_{i=0}^{n-1}\mathbb{E}_{\pi_1}[ d_{1}]}\\
	&=\mathbb{E}_{\pi_1} [a_i]+\frac{1}{2}d_1\\
	&=(\frac{1}{1-p_1}+\frac{1}{2})d_1.
\end{align}
In a similar way, we obtain $\bar{\Delta}_2=(\frac{1}{1-p_2}+\frac{1}{2})d_2$. Hence, we obtain an upper bound $\beta_\text{max}$ of $\beta$, which is given by $$\beta_\text{max}=\min\{\bar{\Delta}_1,\bar{\Delta}_2\}=\min\Big\{\frac{1}{1-p_1}+\frac{1}{2})d_1,(\frac{1}{1-p_1}+\frac{1}{2})d_2\Big\}.$$

\section{Proof of Theorem \ref{CONVEXITY}}
\label{convex_prop}
(i) We first obtain expression of average cost in \eqref{cost2} with aid of state transition diagram.
The state transition diagram under the policy in \eqref{type2} is given in Fig. \ref{state_trans_thres}. 
\begin{figure}
	\centering
	\includegraphics[width=0.48\textwidth]{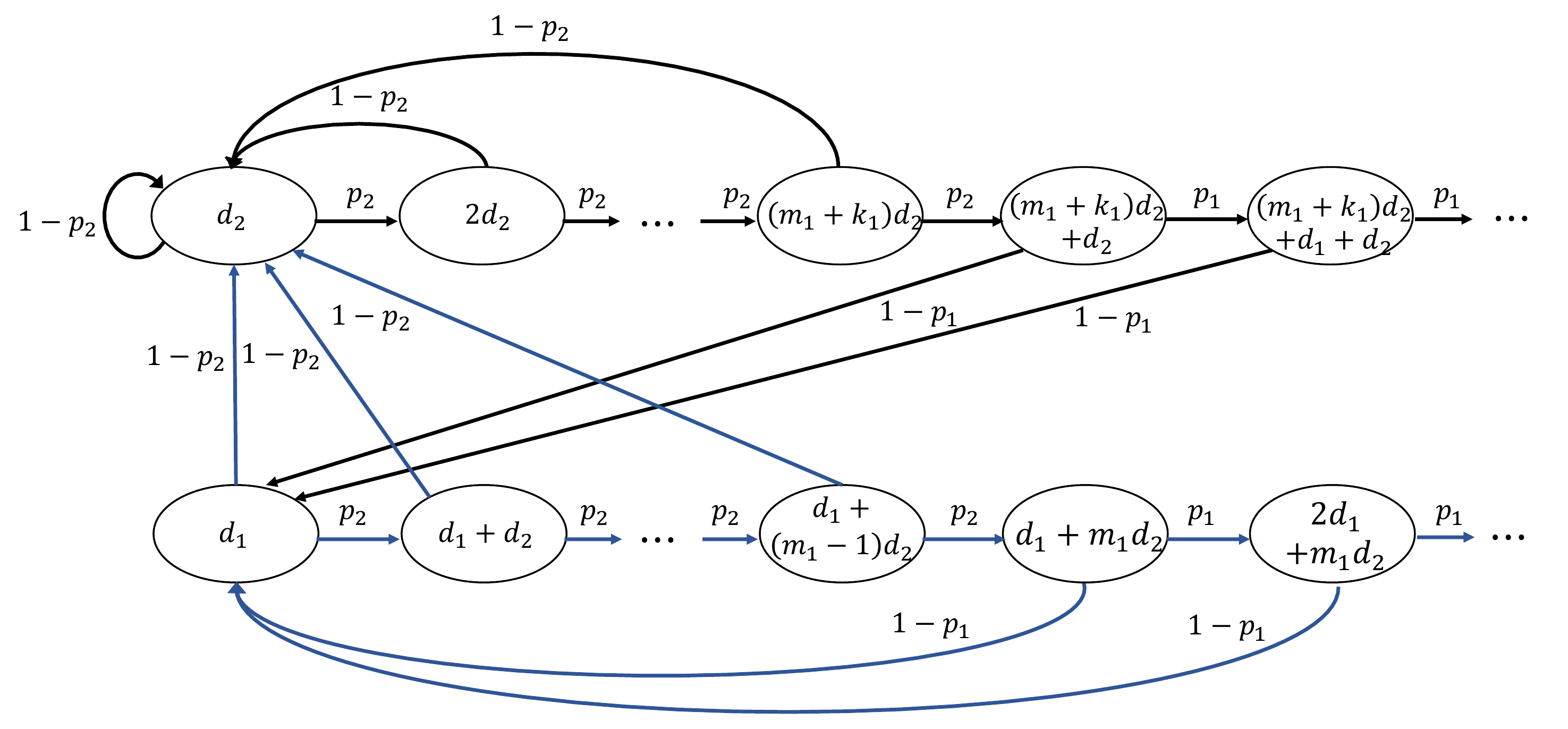}
	\caption{State transition diagram under policy in \eqref{type2}}
	\label{state_trans_thres}
\end{figure}
Define the state steady probabilities $x_l$, $x'_l$, $z_l$ and $z'_l$ under policy in \eqref{type2} as
\begin{align}
	&x_l\triangleq\mathbb{P}(a=d_2+ld_2), &0\leq l\leq m_1+k_1,\\
	&x'_l\triangleq\mathbb{P}(a=d_2+(m_1+k_1)d_2+ld_1), &l\geq 0,\\
	&z_l\triangleq\mathbb{P}(a=d_1+ld_2), &0\leq l\leq m_1,\\
	&z'_l\triangleq\mathbb{P}(a=d_1+m_1d_2+ld_1), &l\geq 0.
\end{align}
Based on the state transition diagram, balance equations can be obtained as follows:
\begin{align}
	& x_0=(1-p_2)\left(\sum_{l=0}^{m_1-1}z_l+\sum_{l=0}^{m_1+k_1-1}x_l\right),\label{e1}\\
	& p_2x_{l}=x_{l+1},	\ \ \ \  l\in \{0,1\cdots m_1+k_1-1\}, \label{e2}\\
	& p_2z_l=z_{l+1}, \ \ \ \  l\in \{0,1\cdots m_1-1\}, \label{e3}\\
	& p_1x'_{l}=x'_{l+1},	\ \ \ \  l\in \{0,1\cdots \}, \label{e4}\\
	& p_1z_l=z'_{l+1},	\ \ \ \  l\in \{0,1\cdots \}. \label{e5}
\end{align}
Solving the equations \eqref{e1}-\eqref{e5}, we obtain the expressions of $x_l$, $x'_l$, $z_l$ and $z'_l$ in terms of $x_0$ as follows:
\begin{align}
	&x_l=p_2^lx_0, \ \ \ \  l\in \{0,1\cdots m_1+k_1\},\label{ee1}\\
	&z_l=\frac{p_2^{m_1+k_1}p_2^l}{1-p_2^{m_1}}x_0,\ \ \ \  l\in \{0,1\cdots m_1\},\label{ee2}\\
	&x'_l=p_2^{m_1+k_1}p_1^lx_0, \ \ \ \  l\in \{0,1\cdots\},\label{ee3}\\
	&z'_l=\frac{p_2^{2m_1+k_1}p_1^l}{1-p_2^{m_1}}x_0,\ \ \ \  l\in \{0,1\cdots\}.,\label{ee4}
\end{align}
Substituting \eqref{ee1}-\eqref{ee4} into $\sum_{l=0}^{m_1+k_1}x_l+\sum_{l=0}^{m_1}z_l+\sum_{l=1}^{\infty}x'_l+\sum_{l=1}^{\infty}z'_l=1$, we obtain $x_0$ as
\begin{align}
	x_0=\frac{(1-p_1)(1-p_2)(1-p_2^{m_1})}{1-p_1+p_2^{m_1}(p_1-1)+p_2^{m_1+k_1}(1-p_2)}.
	\label{x0}
\end{align}

The average cost $J_1(m_1, k_1,\beta)$ is expressed as
\begin{align}
J_1(m_1, k_1,\beta)=&\sum_{l=0}^{m_1+k_1-1}C(d_2+ld_2,d_2)x_l+\sum_{l=0}^{m_1-1}C(d_1+ld_2,d_2)z_l\notag\\
&+\sum_{l=0}^{\infty}C(d_2+(m_1+k_1)d_2+ld_1,d_1)x'_l\notag\\
&+\sum_{l=0}^{\infty}C(d_1+m_1d_2+ld_1,d_1)z'_l.
\label{J}	
\end{align}
where $C(\cdot,\cdot)$ is the cost function defined in \eqref{inst_cost}.
Substitute \eqref{ee1}-\eqref{x0} into \eqref{J}. After some algebraic manipulation and change of variable ($p_2^{m_1}$ is replaced by $y$), we obtain \eqref{cost2}. 

Next, we show that $J_1$ is quasi-convex. By definition of quasi-convex, it suffices to show that its first derivative $\frac{\partial J_1(y, k_1,\beta)}{\partial y}$ with respect to $y$ satisfies at least one of the following conditions \cite{boyd2004convex}:
\begin{itemize}
	\item $\frac{\partial J_1(y, k_1,\beta)}{\partial y}\geq 0$ 
	\item $\frac{\partial J_1(y, k_1,\beta)}{\partial y}\leq 0$ 
	\item there exists a point $y_0\in (0,1]$ such that for $0<y\leq y_0$, $\frac{\partial J_1(y; k_1,\beta)}{\partial y}\leq 0$, and for $1\geq y\geq y_0$, $\frac{\partial J_1(y; k_1,\beta)}{\partial y}\geq 0$
\end{itemize}
After some algebraic manipulation, $\frac{\partial J_1(y; k_1,\beta)}{\partial y}$ is expressed as
\begin{align}
	\frac{\partial J_1(y; k_1,\beta)}{\partial y}=\frac{h(y;k_1,\beta)}{(1-p_1+ry)^2} 
	\label{first_der}
\end{align} 
where $r=-1+p_1+p_2^{k_1}(1-p_2)$ and $h(y;k_1,\beta)$ is given by
\begin{align}
	h(y;k_1,\beta)\!&\triangleq \! \Big(A_1y^2\!+\!\frac{D_2y}{\ln p_2}\!-\!C_1\Big)r\!\notag\\
	&+\!\Big(B_1\!+\!2A_1 y\!+\frac{D_1(1\!+\!\ln y)}{\ln p_2}\Big)(1\!-\!p_2)
\end{align}
Note that the denominator of the \eqref{first_der} is positive. Thus, to show that at least one of the conditions above holds, it suffices to show that $h(y;k_1,\beta)$ satisfies at least one of the following conditions:
\begin{itemize}
	\item B1: $h(y;k_1,\beta)\geq 0$;
	\item B2: $h(y;k_1,\beta)\leq 0$;
	\item B3: $\exists y_0\in (0,1]$ such that for $0<y\leq y_0$, $h(y;k_1,\beta)\leq 0$, and for $1\geq y\geq y_0$, $h(y;k_1,\beta)\geq 0$.
\end{itemize}

In fact, the first derivative of $h(y;k_1,\beta)$ with respect to $y$ is 
\begin{align}
	&\frac{\partial h(y; k_1,\beta)}{\partial y}\notag\\
	=&\underbrace{\left(2(d_2(k_1+1)\!-\!d_1)y-\frac{d_2}{\ln p_2}\right)}_{G(y)}\underbrace{\left(ry+1\!-p_1\right)}_{H(y)}\frac{p_2^{k_1}W}{y}
\end{align}
where $W=d_2(1-p_1)-d_1(1-p_2)$. By condition $d_2(1-p_1)>d_1(1-p_2)$, $W>0$.
Note that since $0<p_1<p_2<1$, $r\leq -1+p_1+(1-p_2)=p_1-p_2<0$. Thus, $H(y)\geq r+1-p_1=p_2^{k_1}(1-p_2)>0$, for $0<y\leq 1$. Hence, $\frac{\partial h(y; k_1,\beta)}{\partial y}$ is positive (negative) if and only if $G(y)$ is positive (negative). Next, we will show our finial result by analyzing two different cases.

\emph{Case 1:} If $d_2(k_1+1)-d_1\geq 0$, then $G(y)\geq-\frac{d_2}{\ln p_2}>0$, for $0<y\leq 1$. In this case, $\frac{\partial h(y; k_1,\beta)}{\partial y}>0$, for $0<y\leq 1$. Thus, $h(y; k_1,\beta)$ is increasing in $y$, for $0<y\leq 1$. Note that since $D_1<0$ (by condition $d_2(1-p_1)>d_1(1-p_2)$) and $p_2<1$, we have $\lim_{y\rightarrow 0} h(y;k_1,\beta)=-\!C_1r+\!\Big(B_1\!+\frac{D_1}{\ln p_2}+\lim_{y\rightarrow 0}\frac{D_1\!\ln y}{\ln p_2}\Big)(1\!-\!p_2)<0$. Thus, if $h(1; k_1,\beta)\leq 0$, then B2 holds; otherwise, B3 holds. 

\emph{Case 2:} If $d_2(k_1+1)-d_1<0$, then $G(y)$ is decreasing in $y$ and $y_1=\frac{d_2}{2(d_2(k_1+1)-d_1)\ln p_2}>0$ is a turning point such that $G(y)>0$ when $y<y_1$ and $G(y)<0$ when $y>y_1$. 

If $y_1\geq 1$, then $G(y)\geq 0$ for $0<y\leq 1$. Thus, $\frac{\partial h(y; k_1,\beta)}{\partial y}\geq 0$, which implies that $h(y; k_1,\beta)$ is increasing in $y$ for $0<y\leq 1$. Hence, B2 or B3 holds as explained in case 1.

If $y_1<1$, then $G(y)>0$ when $y<y_1$ and $G(y)<0$ when $1>y>y_1$, which implies $h(y; k_1,\beta)$ first increases and then decreases for $0<y\leq 1$. 
We claim that if $y_1<1$, then $h(1; k_1,\beta)\geq 0$ for $\beta\in [\beta_\text{min}, \beta_\text{max}]$ and $k_1 \in\mathcal{K}_1$. Recall that $\lim_{y\rightarrow 0} h(y;k_1,\beta)<0$. With this, the claim implies that $h(y;k_1,\beta)$ starts with some negative value and increases to zero at some point $y'$ and after $h(y;k_1,\beta)$ becomes positive, it will keep positive for $y'<y\leq 1$ (B3 holds). It remains to show that our claim holds.

Since $y_1=\frac{d_2}{2(d_2(k_1+1)-d_1)\ln p_2}<1$, we have  
\begin{align}
	\frac{d_1}{d_2}>k_1+1-\frac{1}{2\ln p_2}
	\label{cond}
\end{align}

After some algebraic manipulation and simplification, $h(1; k_1,\beta)$ is expressed as
\begin{align}
	&h(1; k_1,\beta)\notag\\
	=&p_2^{2k_1}W\left(d_2(k_1+1)-d_1-\frac{d_2}{\ln p_2}\right)(1-p_2)\notag\\
	&+p_2^{k_1}(1-p_2)(1-p_1)\left((\frac{1}{1-p_1}+0.5)d_1^2-(\frac{1}{1-p_2}+0.5)d_2^2\right)\notag\\
	&+p_2^{k_1}(d_2-d_1)(1-p_1)(1-p_2)\beta
\end{align}
Since $p_2^{k_1}(d_2-d_1)(1-p_1)(1-p_2)<0$, $h(1; k_1,\beta)$ is decreasing in $\beta$ and thus $h(1; k_1,\beta)\geq h(1; k_1,\beta_\text{max})\geq h(1; k_1,(\frac{1}{1-p_1}+0.5)d_1)$. It remains to show that $h(1; k_1,(\frac{1}{1-p_1}+0.5)d_1)\geq 0$. It is equivalent to show that $\frac{h(1; k_1,(\frac{1}{1-p_1}+0.5)d_1)}{d_2^2}\geq 0$ since $d_2^2\geq 0$.
Let $z=\frac{d_1}{d_2}$, substitute this into $\frac{h(1; k_1,(\frac{1}{1-p_1}+0.5)d_1)}{d_2^2}$ and obtain a function of $z$ denoted by $w(z)$. By \eqref{cond}, the first derivative of $w(z)$ satisfies
\begin{align}
	w'(z)=&2p_2^{2k_1}(1-p_2)^2z+p_2^{k_1}(1-p_2)(\frac{1}{1-p_1}+0.5)(1-p_1)\notag\\
		\displaybreak
	+&p_2^{2k_1}(1-p_2)\left(\frac{1-p_2}{\ln p_2}-(1-p_2)(k_1+1)-(1-p_1)\right)\\
	\geq & w'(k_1+1-\frac{1}{2\ln p_2})\\
	=& p_2^{k_1}(1-p_2)\left(p_2^{k_1}(1-p_2)(k_1+1)+\frac{1-p_1}{2})\right)\notag\\
	+&p_2^{k_1}(1-p_2)\left(1-p_2^{k_1}(1-p_1)\right)\\
	>&0
	\label{tt2}
\end{align}
The \eqref{tt2} holds since $0<p_1<p_2<1$. Thus, $w(z)$ is increasing and we have
\begin{align}
	w(z)\geq & w(k_1+1-\frac{1}{2\ln p_2})\\
	=&\underbrace{p_2^{k_1}(1-p_2)}_{>0}\underbrace{\left(p_2^{k_1}(1-p_2)+\ln p_2\right)}_{ s_1(p_2)}\underbrace{\left(\frac{k_1}{2\ln p_2}-\frac{1}{4(\ln p_2)^2}\right)}_{<0}\notag\\
	+&\frac{p_2^{k_1}}{2\ln p_2}\delta(p_1,p_2,k_1)
\end{align}
where 
\begin{align}
	\delta(p_1,p_2,k_1)=&(p_1-p_2)(1-p_2)p_2^{k_1}+2(p_1-p_2)\ln p_2\notag\\
	&+2k_1(1-0.5p_2)(1-p_2)\ln p_2+(1-p_2)(0.5p_1-1)\notag
\end{align}
Note that $s_1(p_2)\leq 1-p_2+\ln p_2<0$. Thus, to show that $w(z)\geq 0$ holds, we only need to show that $\delta(p_1,p_2,k_1)\leq 0$ holds. Actually, 
\begin{align}
	&\frac{\partial \delta(p_1,p_2,k_1)}{\partial p_2}\notag\\
	&=\underbrace{(p_1-p_2)(1-p_2)p_2^{k_1}\ln p_2+k_1(2-p_1)(\frac{1}{p_2}-1-\ln p_2)}_{>0 \text{ by } 0<p_1<p_2<1}\notag\\
	&+\underbrace{p_2^{k_1}(2p_2-1-p_1)+\frac{2p_1}{p_2}-1-\frac{p_1}{2}-2\ln p_2}_{\theta(p_1,p_2,k_1)}\\&\geq 0
	\label{tt3}
\end{align}
where the \eqref{tt3} holds since $\theta(p_1,\!p_2,\!k_1)$ decreases with $p_2$ and $\theta(p_1,\!1,\!k_1)=0.5p_1>0$. Actually, 
\begin{align}
	\frac{\partial \theta(p_1,p_2,k_1)}{\partial p_2}&=(2p_2-1-p_1)p_2^{k_1}\ln p_2+2p_2^{k_1}-2\frac{p_1}{p_2^2}-\frac{2}{p_2}\label{q0}\\
	& <(p_2-1)p_2^{k_1}\ln p_2+2p_2^{k_1}-2\frac{p_1}{p_2^2}-\frac{2}{p_2}\label{q1}\\
	&\leq-(p_2-1)p_2^{k_1}\frac{1}{2p_2}+2p_2^{k_1}-2\frac{p_1}{p_2^2}-\frac{2}{p_2}\label{q2}\\
    &\leq -(p_2-1)\frac{1}{2p_2}+2-2\frac{p_1}{p_2^2}-\frac{2}{p_2} \label{q3}\\
    &<0
\end{align}
where \eqref{q1} holds since $p_2-p_1>0$ and $\ln p_2 <0$; \eqref{q2} holds by $s_2(x)\triangleq\frac{1}{x}+2\ln x\geq0$ for $0<x\leq 1$; \eqref{q3} holds since $p_2^{k_1}\leq 1$; \eqref{q3} holds since $p_2< 1$. In particular, the first derivative of $s_2$ is $s_2'(x)=-\frac{1}{x^2}+\frac{2}{x}$. Thus, $s_2$ decreases when $x\leq 0.5$ (since $s_2'(x)\leq 0$ when $x\leq 0.5$) and then increases when $x\geq 0.5$ (since $s_2'(x)\geq0$ when $x\geq 0.5$) . Thus, $s_2(x)\geq s_2(0.5)=0.2213>0$.
By \eqref{tt3}, $\delta$ increases with $p_2$ and $\delta(p_1,p_2,k_1)\leq \delta(p_1,1,k_1)=0$. This completes the proof of part (i).

(ii) If $d_1(1-p_2)\geq d_2(1-p_1)$, the optimal policy is of threshold-type in \eqref{type1}. Similar to part (i), with aid of state transition diagram, we obtain the expression of average cost $J_2$ under policy \eqref{type1} as 
	\begin{equation}
		J_2(m_2, k_2,\beta)=\frac{A_2p_1^{2m_2}+B_2p_1^{m_2}+C_2+D_2m_2p_1^{m_2}}{1-p_2+(1-p_1-p_1^{k_2}(1-p_2))p_1^{m_2}}.
		\label{J2}
	\end{equation}
In particular, the state transition diagram under policy in \eqref{type1} is given in Fig. \ref{state_trans_thres2}. 
\begin{figure}
	\centering
	\includegraphics[width=0.48\textwidth]{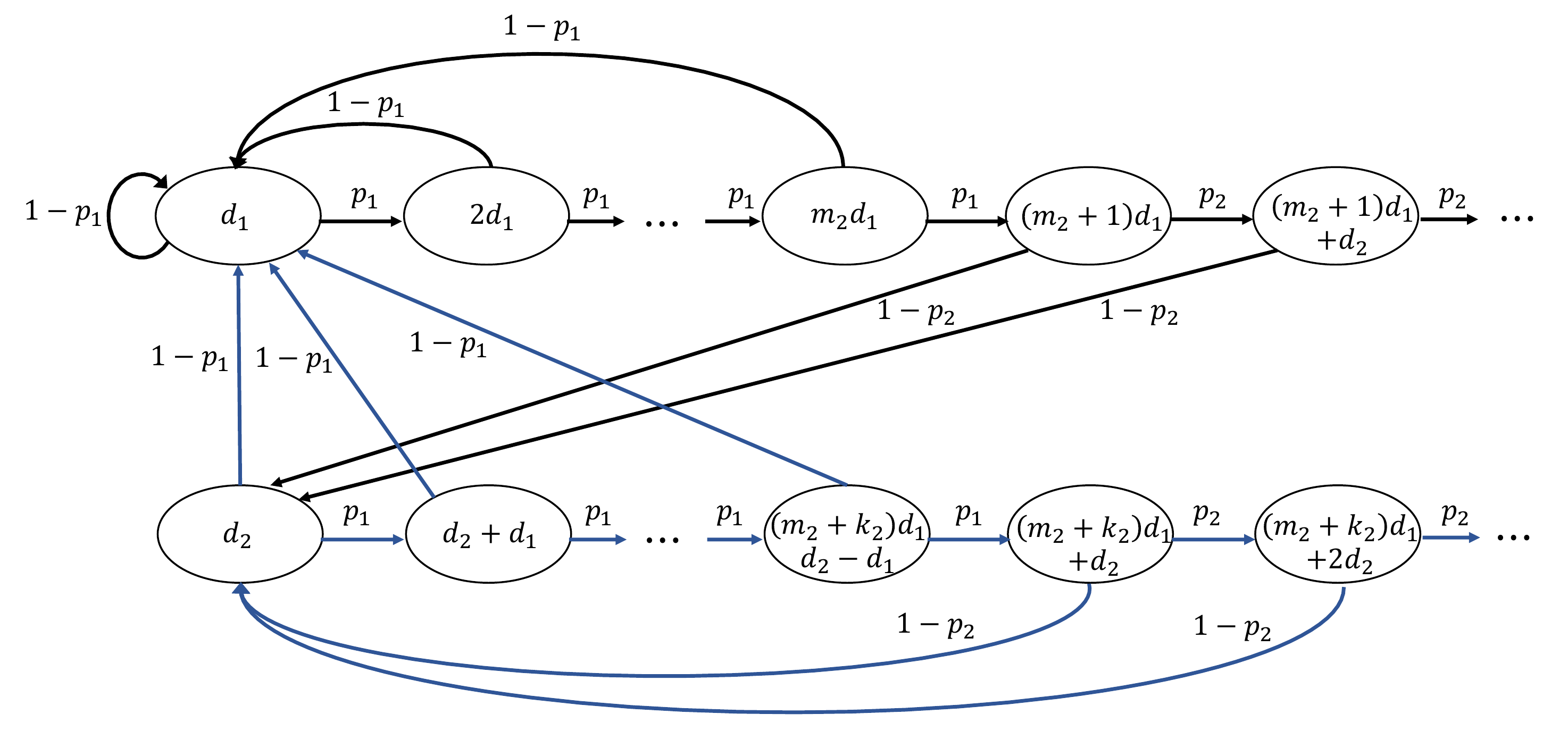}
	\caption{State transition diagram under policy in \eqref{type1}}
	\label{state_trans_thres2}
\end{figure}
Define the steady state probabilities $x_l$, $x'_l$, $y_l$ and $y'_l$ under policy in \eqref{type1} as
\begin{align}
	&x_l\triangleq\mathbb{P}(a=d_1+ld_1), &0\leq l\leq m_2,\\
	&x'_l\triangleq\mathbb{P}(a=d_1+m_2d_1+ld_2), &l\geq 0,\\
	&z_l\triangleq\mathbb{P}(a=d_2+ld_1), &0\leq l\leq m_2+k_2,\\
	&z'_l\triangleq\mathbb{P}(a=d_2+(m_2+k_2)d_1+ld_2), &l\geq 0.
	\end{align}
Based on the state transition diagram, balance equations can be obtained as follows:
\begin{align}
	& x_0=(1-p_1)\left(\sum_{l=0}^{m_2-1}x_l+\sum_{l=0}^{m_2+k_2-1}z_l\right),\label{e12}\\
	& p_1x_{l}=x_{l+1},	\ \ \ \  l\in \{0,1\cdots m_2-1\}, \label{e22}\\
	& p_1z_{l}=z_{l+1}, \ \ \ \  l\in \{0,1\cdots m_2+k_2-1\}, \label{e32}\\
	& p_2x'_{l}=x'_{l+1},	\ \ \ \  l\in \{0,1\cdots \}, \label{e42}\\
	& p_2z'_{l}=z'_{l+1},	\ \ \ \  l\in \{0,1\cdots \}. \label{e52}
\end{align}
Solving the equations \eqref{e12}-\eqref{e52}, we obtain the expressions of $x_l$, $x'_l$, $z_l$ and $z'_l$ in terms of $x_0$ as follows:
\begin{align}
	&x_l=p_1^lx_0, \ \ \ \  l\in \{0,1\cdots m_2\},\label{ee12}\\
	&z_l=\frac{p_1^{m_2}p_1^l}{1-p_1^{m_2+k_2}}x_0,\ \ \ \  l\in \{0,1\cdots m_2+k_2\},\label{ee22}\\
	&x'_l=p_1^{m_2}p_2^lx_0, \ \ \ \  l\in \{0,1\cdots\},\label{ee32}\\
	&z'_l=\frac{p_1^{2m_2+k_2}p_2^l}{1-p_1^{m_2+k_2}}x_0,\ \ \ \  l\in \{0,1\cdots\}.,\label{ee42}
\end{align}
Substituting \eqref{ee12}-\eqref{ee42} into $\sum_{l=0}^{m_2-1}x_l+\sum_{l=0}^{m_2+k_2-1}z_l+\sum_{l=0}^{\infty}x'_l+\sum_{l=0}^{\infty}z'_l=1$, we obtain $x_0$ as
\begin{align}
	x_0=\frac{(1-p_1)(1-p_2)(1-p_1^{m_2+k_2})}{1-p_2+p_1^{m_2+k_2}(p_2-1)+p_1^{m_2}(1-p_1)}.
	\label{x02}
\end{align}

The average cost $J_2(m_2, k_2,\beta)$ is expressed as
\begin{align}
J_2(m_2, k_2,\beta)=&\sum_{l=0}^{m_2-1}C(d_1+ld_1,d_1)x_l+\sum_{l=0}^{m_2+k_2-1}C(d_2+ld_1,d_1)z_l\notag\\
&+\sum_{l=0}^{\infty}C(d_1+m_2d_1+ld_2,d_2)x'_l\notag\\
&+\sum_{l=0}^{\infty}C(d_2+(m_2+k_2)d_1+ld_2,d_2)z'_l.
\label{J_2}	
\end{align}
Substitute \eqref{ee12}-\eqref{x02} into \eqref{J_2}. After some algebraic manipulation and simplification, we obtain \eqref{J2}.
	
	Next, we show that the optimal policy that solves \eqref{target1} chooses $u=2$ at each transmission opportunity, which is equivalent to show that $m_2=0$ and $n_2=0$. In particular, (1) we show that for any $m_2$, $J_2(m_2,k_2,\beta)$ is non-decreasing in $k_2$, which implies that optimal $k_2$ equals zero; (2) we show that $J_2(m_2, 0,\beta)$ is non-decreasing in $m_2$. This implies that optimal $m_2$ equals zero. Since $n_2=m_2+k_2$, then optimal $n_2$ equals to zero and thus the optimal average cost is $J_2(0,0,\beta)$, which will complete the proof of part (ii). The proof is specified as follows:
	
	(1) We show that any $m_2$, $J_2(m_2,k_2,\beta)$ is non-decreasing in $k_2$. Since $k_2\in \mathcal{K}_2=\{0,1\}$, it suffices to show that $J_2(m_2,1,\beta)-J_2(m_2,0,\beta)\geq 0$.
	    Using the expression of $J_2$ in \eqref{J2}, after some algebraic manipulation, we get
	    \begin{align}
	    	J_2(m_2,1,\beta)-J_2(m_2,0,\beta)=\frac{p_1^{2m_2}g_1(m_2,\beta)}{F_1(m_2)}
	    \end{align}
	    where
	    \begin{align}
	    	&g_1(m_2,\beta)\notag\\
	    	=&p_1^{m_2}\Big(\left(1\!-\!p_1\right)^2d_2-d_1(1-p_1)+p_1d_1(1\!-\!p_2)\Big)E\notag\\
	    	&+m_2d_1(1-p_1)(1-p_2)E\notag\\
	    	&+(1-p_1)^2(1-p_2)\big((d_2-d_1)\beta+d_1\beta_1-d_2\beta_2\big)\notag\\
	    	&+\big(\left(1\!-\!p_1\right)d_2-d_1\big)(1-p_2)E+(d_1-d_2)(1-p_1)(1-p_2)E
	    	\label{g1_m_beta}
	    \end{align}
	    where $\beta_1=\frac{1}{1-p_1}d_1+0.5d_1$, $\beta_2=\frac{1}{1-p_2}d_2+0.5d_2$, and $E=d_1(1-p_2)-d_2(1-p_1)$; and $F_1(m_2)$ is given by
	    \begin{align}
	    	F_1(m_2)=\big(1\!-\!p_2\!+\!p_1^{m_2}(1\!-\!2p_1+p_1p_2)\big)\big(1\!-\!p_2+p_1^{m_2}(p_2\!-\!p_1)\big)
	    \end{align}
	    Since $p_1<p_2<1$, we have that $1-2p_1+p_1p_2=1-p_1-p_1(1-p_2)\geq 1-p_1-(1-p_2)\geq 0$. It means that $F_1$ is a product of two non-negative expressions and thus we have $F_1(m_2)\geq 0$, for all $m_2$. Hence, to show the result, we only need to show that $g_1(m_2,\beta)\geq 0$ holds. In particular, we have 
	    \begin{align}
	    	&\left(1\!-\!p_1\right)^2d_2-d_1(1-p_1)+p_1d_1(1\!-\!p_2)
	    	\notag\\
	    	\leq & \left(1\!-\!p_1\right)^2d_2-d_1(1-p_2)+p_1d_1(1\!-\!p_2)\label{ineq1}\\
	    	= & \left(1\!-\!p_1\right)(d_2\left(1\!-\!p_1\right)-d_1\left(1\!-\!p_2\right))\\
	    	\leq & \label{ineq2} 0
	    \end{align}
	    where \eqref{ineq1} holds since $0<p_1<p_2<1$, and \eqref{ineq2} holds due to condition $d_1(1-p_2)\geq d_2(1-p_1)$ and $p_1<1$.	 In addition, by condition $d_1(1-p_2)\geq d_2(1-p_1)$, we get $E\geq 0$. Thus, the first term in \eqref{g1_m_beta} satisfies
	    \begin{align}
	    	&p_1^{m_2}\Big(\left(1\!-\!p_1\right)^2d_2-d_1(1-p_1)+p_1d_1(1\!-\!p_2)\Big)E\notag\\
	    	\geq & \Big(\left(1\!-\!p_1\right)^2d_2-d_1(1-p_1)+p_1d_1(1\!-\!p_2)\Big)E
	    	\label{first_term}
	    \end{align}
	    Moreover, the second term in \eqref{g1_m_beta} satisfies
	    \begin{align}
	    	m_2d_1(1-p_1)(1-p_2)E\geq 0
	    	\label{second_term}
	    \end{align} 
	Recall that $\beta_1=\frac{1}{1-p_1}d_1+0.5d_1$, $\beta_2=\frac{1}{1-p_2}d_2+0.5d_2$ and $\beta_{max}=\min\{\beta_1,\beta_2\}$. Since $d_1>d_2$ and $d_1(1-p_2)\geq d_2(1-p_1)$, we have $\beta_1\geq\beta_2$. Hence, $\beta\leq\beta_{max}=\min\{\beta_1,\beta_2\}=\beta_2$. Moreover, since $d_2<d_1$ and $0<p_1<p_2<1$, we have $(1-p_1)^2(1-p_2)(d_2-d_1)<0$. Hence, the third term in \eqref{g1_m_beta} satisfies
		 \begin{align}
	    	(1-p_1)^2(1-p_2)(d_2-d_1)\beta\geq (1-p_1)^2(1-p_2)(d_2-d_1)\beta_2
	    	\label{third_term}
	    \end{align}
Substitute \eqref{first_term}, \eqref{second_term} and \eqref{third_term} into \eqref{g1_m_beta}, after some algebraic manipulation, we get
\begin{align}
	g_1(m_2,\beta)
	\geq & Ed_1(1-p_1)(1-p_2)+E(1-p_1)^2d_2\notag\\
	&+0.5d_1(1-p_1)^2(1-p_2)(d_1-d_2)\\
	\geq & 0\label{ineq3}
\end{align}
where \eqref{ineq3} holds by $E\geq 0$, $d_1>d_2>0$ and $0<p_1<p_2<1$.

(2) We show that $J_2(m_2, 0,\beta)$ is non-decreasing in $m_2$ by showing $J_2(m_2+1, 0,\beta)-J_2(m_2, 0,\beta)\geq 0$, for all $m_2\in\mathbb{N}$. Using the expression of $J_2$ in \eqref{J2}, after some algebraic manipulation, we get 
\begin{align}
	&J_2(m_2+1, 0,\beta)-J_2(m_2, 0,\beta)=\frac{g_2(m_2,\beta)p_1^{m_2}}{F_2(m_2)}
\end{align}
where 
\begin{align}
	&g_2(m_2,\beta)\notag\\
	=&p_1^{2m_2+1}(d_2-d_1)(1-p_1)(p_2-p_1)E-p_1^{m_2+1}d_1(p_2-p_1)E\notag\\
	&+p_1^{m_2}(d_2-d_1)(1-p_1^2)(1-p_2)E+m_2d_1(1-p_1)(1-p_2)E\notag\\
	&+(1-p_1)^2(1-p_2)\big((d_2-d_1)\beta+d_1\beta_1-d_2\beta_2\big)\notag\\
	&+(d_1-d_2)(1-p_1)(1-p_2)E-d_1p_1(1-p_2)E,
	\label{g2_m_beta}
\end{align}
and $F_2(m_2)$ is given by
	    \begin{align}
	    	F_2(m_2)=\big(1\!-\!p_2\!+\!p_1^{m_2+1}(p_2-p_1)\big)\big(1\!-\!p_2+p_1^{m_2}(p_2\!-\!p_1)\big).
	    \end{align}
	    Since $p_1<p_2<1$, $F_2$ is a product of two non-negative terms and thus $F_2(m_2)\geq 0$, for all $m_2$. Hence, to show the result, we only need to show that $g_2(m_2,\beta)\geq 0$ holds. In particular, since $0<p_1<p_2<1$ and $d_1>d_2$, the first four terms in \eqref{g2_m_beta} satisfy 
\begin{align}
	&p_1^{2m_2+1}(d_2-d_1)(1-p_1)(p_2-p_1)E-p_1^{m_2+1}d_1(p_2-p_1)E\notag\\
	&+p_1^{m_2}(d_2-d_1)(1-p_1^2)(1-p_2)E+m_2d_1(1-p_1)(1-p_2)E\notag\\
	\geq & p_1(d_2-d_1)(1-p_1)(p_2-p_1)E-p_1d_1(p_2-p_1)E\notag\\
	&+(d_2-d_1)(1-p_1^2)(1-p_2)E
	\label{ineq5}
\end{align}	    
Substitute \eqref{third_term} and \eqref{ineq5} into \eqref{g2_m_beta}, after some algebraic manipulation, we get	 
\begin{align}
	g_2(m_2,\beta)\geq & (1-p_1)^2\big(d_1(1-p_1)+d_2p_1\big)E\notag\\
	&+0.5(1-p_1)^2(1-p_2)(d_1-d_2)d_1\\
	\geq & 0
	\label{ineq6}
\end{align}   
where \eqref{ineq6} holds by $E\geq 0$, $d_1>d_2>0$ and $0<p_1<p_2<1$.

\end{document}